%% file: Paper_Filters.tex
\documentclass[a4paper]{IEEEtran}
\usepackage{graphicx,ifthen,stfloats}
\usepackage{psfrag,amssymb,amsthm}
\usepackage[cmex10]{amsmath}
\usepackage{cite,color,pst-plot,stfloats,pst-sigsys}
\usepackage{pstricks-add}

\newtheorem{thm}{Theorem}
\newtheorem{lem}{Lemma}
\newtheorem{cor}{Corollary}

\theoremstyle{definition}
\newtheorem{definition}{Definition}
\newtheorem{ass}{Assumption}

\newcounter{excnt}

\title{Information Loss and Anti-Aliasing Filters in Multirate Systems}

\author{Bernhard C. Geiger,~\IEEEmembership{Student Member, IEEE}, and Gernot Kubin,~\IEEEmembership{Member, IEEE}%
\thanks{Bernhard C. Geiger (geiger@ieee.org) and Gernot Kubin are with the Signal Processing and Speech Communication Laboratory, Graz University of Technology.}%
\thanks{Parts of this work have been presented at the IEEE Forum on Signal Processing for RF-Systems 2013 and at the 2014 Int. Z\"urich Seminar on Communications.}}

\begin{document}
\newcounter{myTempCnt}

% \input{../../Blocks/abbrevations_processing.tex}
\input{abbrevations_processing.tex}

\renewcommand{\mutrate}[1]{\overbar{I}({#1})}
\renewcommand{\lossrate}[2][\empty]{\ifthenelse{\equal{#1}{\empty}}{\overbar{L}(\mathbf{#2})}{\overbar{L}_{#1}(#2)}}
\renewcommand{\derate}[1]{\overbar{h}(#1)}
\renewcommand{\relLossrate}[1]{\overbar{l}(#1)}
\newcommand{\blocked}[2]{#1^{\mathit{(#2)}}}
\newcommand{\Mblocked}[1]{\blocked{#1}{M}}
\newcommand{\XvecMb}{\Mblocked{\Xvec}}
\newcommand{\XMb}{\Mblocked{X}}

\maketitle

\begin{abstract}
This work investigates the information loss in a decimation system, i.e., in a downsampler preceded by an anti-aliasing filter. It is shown that, without a specific signal model in mind, the anti-aliasing filter cannot reduce information loss, while, e.g., for a simple signal-plus-noise model it can. For the Gaussian case, the optimal anti-aliasing filter is shown to coincide with the one obtained from energetic considerations. For a non-Gaussian signal corrupted by Gaussian noise, the Gaussian assumption yields an upper bound on the information loss, justifying filter design principles based on second-order statistics from an information-theoretic point-of-view.
\end{abstract}

\section{Introduction}
Multi-rate systems are ubiquitously used in digital systems to increase (upsample) or decrease (downsample) the rate at which a signal is processed. Especially downsampling is a critical operation since it can introduce aliasing, like sampling, and thus can cause information loss. Standard textbooks on signal processing deal with this issue by recommending an anti-aliasing filter prior to downsampling -- resulting in a cascade which is commonly known as a decimator~\cite[Ch.~4.6]{Oppenheim_Discrete3}. In these books, this anti-aliasing filter is usually an ideal low-pass filter with a cut-off frequency of $\pi/M$, for an $M$-fold decimation system (cf.~Fig.~\ref{fig:downsampling}). Unser~\cite{Unser_OptimalFilters} showed that this choice is optimal in terms of the mean-squared reconstruction error (MSE) only if the input process is such that the passband portion of its power spectral density (PSD) exceeds all aliased components.  Similarly, as it was shown by Tsatsanis and Giannakis~\cite{Tsatsanis_PCFilterBank}, the filter minimizing the MSE is piecewise constant, $M$-aliasing-free (i.e., the aliased components of the $M$-fold downsampled frequency response do not overlap), and has a passband depending on the PSD of the input process. Specifically, the filter which permits most of the energy to pass aliasing-free is optimal in the MSE sense.

In this paper we consider a design objective vastly different from the MSE: information. The fact that information, compared to energy, can yield more successful system designs has long been recognized, e.g., for (non-linear) adaptive filters~\cite{Principe_ITLearning} or for state estimation using linear filters~\cite{Galdos_InfoFilterDesign}. Mutual information has been used as a design objective for transceiver filter design, too: In~\cite{Al-Dhahir_BlockTransmission}, Al-Dhahir et al. derived a sub-optimal block transmission filter whose output approximates the optimal input statistics of a dispersive, noisy channel. Scaglione et al.~\cite{Scaglione_FBs} later showed that the resulting non-stationarity of the channel input can be achieved by an FIR filter bank, both for FIR and ARMA channels. Recently, Chen et al.~\cite{Chen_NyquistShannon} derived the capacity of sub-Nyquist sampled additive Gaussian noise channels for various sampling mechanisms: sampling after a filter, a filterbank, and a modulated filterbank. They showed that the capacity-maximizing sampling filter is piecewise constant and both maximizes the signal-to-noise ratio and minimizes the MSE of the reconstructed signal, thus building a bridge between information-theoretic and energetic filter design. All these works, however, consider either a signal-plus-noise model or assume that all processes are Gaussian.

We extend the existing literature in three different aspects: First, we present results for the case where no signal model is available, other than the PSD of the Gaussian input process (Section~\ref{sec:relativeLoss}). Second, we derive the optimal filter for a signal model in which the Gaussian filter input is correlated with a Gaussian signal process representing relevant information (Section~\ref{sec:relevantGauss}). And finally, we consider a signal-plus-Gaussian-noise scenario, where we assume that the signal is non-Gaussian (Section~\ref{sec:relevantGeneral}).

% In~\cite{Al-Dhahir_BlockTransmission}, for block transmission, the optimal input covariance matrix for a dispersive, noisy channel was shown to be non-stationary; the authors derived a sub-optimal transmit filter whose output approximates the optimal input statistics. 
%
% Scaglione et al.~\cite{Scaglione_FBs} later showed that the resulting non-stationarity of the channel input can be achieved by an FIR filter bank, both for FIR and ARMA channels. Interestingly, both the zero-forcing and the MMSE equalizer filter banks satisfy the optimality conditions in terms of the block mutual information.
%
% Recently, Chen et al.~\cite{Chen_NyquistShannon} derived the capacity of sub-Nyquist sampled additive Gaussian noise channels for various sampling mechanisms: sampling after a filter, a filterbank, and a modulated filterbank. They also derived, for a given channel transfer function, the optimal filter(s); in the first case, sampling after a filter, their result is remarkably similar to ours presented in Theorem~\ref{thm:energycompaction} below. However, while in~\cite{Chen_NyquistShannon} the capacity of the channel is considered, this work presents results for given signal and noise PSDs (generally different from the capacity-achieving solutions).
% 
% Finally, quantifying information relieves us from having to specify a reconstruction procedure: The information lost in the decimation system is independent from signal reconstruction, therefore a separate design of these two system components should be possible.

\begin{figure}[t]
\centering
 \begin{pspicture}[showgrid=false](-8,-2.5)(-2,-1)
  \psset{style=RoundCorners,style=Arrow,gratioWh=1.2}
	\pssignal(-8,-2){X}{$\Xvec$}
	\psfblock(-6.25,-2){H0}{$H$}
	\psdsampler(-3.75,-2){DO1}{M}
	\pssignal(-2,-2){y2}{$\Yvec$}
	\nclist{ncline}[naput]{X,H0,DO1 $\tilde{\Xvec}$,y2}
 \end{pspicture}
\caption{Decimation system consisting of a linear filter $H$ and an $M$-fold downsampler.}
\label{fig:downsampling}
\end{figure}
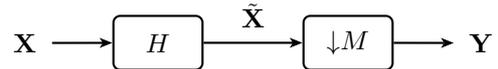

Our first result is surprising: Given mild assumptions on the input process of the decimation system, the information loss can be bounded \emph{independently} of the anti-aliasing filter (see Section~\ref{sec:relativeLoss}). The reason is that, without a specific signal model, every bit of the input process is treated equivalently, regardless of the amount of energy by which it is represented. In order to remedy this counter-intuitivity, Section~\ref{sec:relevantGauss} considers Gaussian processes with a specific signal model in mind: The input to the decimation system is correlated with a relevant data signal. A data signal corrupted by Gaussian noise is a special case of this scenario, thus connecting to the analysis of additive Gaussian noise channels in~\cite{Chen_NyquistShannon}. The optimal filter is shown to be piecewise constant and conceptually similar to those derived in~\cite{Chen_NyquistShannon,Tsatsanis_PCFilterBank}. Since in most cases the Gaussian assumption is too restrictive, in Section~\ref{sec:relevantGeneral} we let the decimator input be an arbitrarily distributed data signal corrupted by Gaussian noise. Following the approach of Plumbley in~\cite{Plumbley_TN}, we prove that the Gaussian assumption for the signal process yields an upper bound on the information loss in the general case. In other words, designing a filter based on the PSDs of the signal and noise processes guarantees a bounded information loss in the decimation system. This justifies filter design based on second-order statistics, i.e., on energetic considerations, also from an information-theoretic perspective. In Section~\ref{sec:example} we illustrate our results in a simple toy example. The problem of designing optimal FIR filters for the decimation system, as it appears in the toy example, is briefly discussed in Section~\ref{sec:FIR}, which also contains an outlook to future work.

\section{Preliminaries and Notation}\label{sec:preliminaries}
Throughout this work we adopt the following notation: $\Zvec$ is a real-valued random process, whose $n$-th sample is the random variable (RV) $Z_n$. Let $Z_\mathbb{J}:=\{Z_i:i\in\mathbb{J}\}$ and we abbreviate $Z_i^j:=\{Z_i, Z_{i+1},\dots,Z_j\}$. The differential entropy~\cite[Ch.~8]{Cover_Information2} and the R\'{e}nyi information dimension~\cite{Renyi_InfoDim} of $Z_\mathbb{J}$ are $\diffent{Z_\mathbb{J}}$ and $\infodim{Z_\mathbb{J}}$, respectively, provided these quantities exist and are finite. Finally, we define the $M$-fold blocking $\Zvec^{(M)}$ of $\Zvec$ as the sequence of $M$-dimensional RVs $Z^{(M)}_1:=Z_1^M$, $Z^{(M)}_2:=Z_{M+1}^{2M}$, and so on. Since the systems used in this work are static, i.e., described by a function, we abuse notation and understand $g(Z_\mathbb{J})$ and $g(\Zvec)$ as the function applied coordinate-wise or sample-wise, respectively.

In this work, we often consider a process $\Zvec$ satisfying
\begin{ass}\label{ass:processes}
 $\Zvec$ is stationary, has finite variance, finite marginal differential entropy $\diffent{Z_n}$, and finite differential entropy rate
\begin{equation}
 \derate{\Zvec} := \limn \frac{1}{n} \diffent{Z_1^n} = \limn \diffent{Z_n|Z_1^{n-1}}.
\end{equation}
\end{ass}

\begin{lem}[Finite differential entropy (rate) and information dimension]\label{lem:DEInfoDim}
 Let $\Zvec$ be a stationary stochastic process satisfying Assumption~\ref{ass:processes}. Then, for every finite set $\LetterSet{J}$, $\infodim{Z_\LetterSet{J}}=\card{\LetterSet{J}}$.
\end{lem}

\begin{IEEEproof}
See Appendix~\ref{proof:DEInfoDim}.
\end{IEEEproof}

As another direct consequence of Assumption~\ref{ass:processes}, the mutual information rate with a process $\Wvec$ jointly stationary with $\Zvec$ exists and equals~\cite[Thm.~8.3]{Gray_Entropy}
\begin{equation}
 \mutrate{\Zvec;\Wvec} := \limn \frac{1}{n} \mutinf{Z_1^n;W_1^n}.
\end{equation}

To measure the rate of information loss in a deterministic system, we introduce
\begin{definition}[Relative Information Loss Rate]\label{def:relativerate}
The relative information loss rate induced by the function $g$ is
 \begin{equation}
 \relLossrate{\Zvec\to g(\Zvec)} := \limn \relLoss{(Z_1^n\to g(Z_1^n)}=\limn \frac{\infodim{Z_1^n|g(Z_1^n)}}{\infodim{Z_1^n}}~\label{eq:relativerate}
\end{equation}
provided the quantity on the right exists.
\end{definition}

This definition is an extension of the \emph{relative information loss} $\relLoss{Z\to g(Z)}$, as defined in~\cite{Geiger_RILPCA_arXiv}, to stochastic processes. Roughly speaking, $\relLoss{Z\to g(Z)}$ captures the \emph{percentage} of information lost by applying the function $g$ to the RV $Z$. That the relative information loss is related to the (conditional) information dimension was also observed in~\cite{Geiger_RILPCA_arXiv}, where the second equality in~\eqref{eq:relativerate} was proved.

Clearly, in an invertible system no information is lost. One drawback of Definition~\ref{def:relativerate} is that it is only defined for static systems. Hence, it will be necessary to abuse notation by presenting
\begin{ass}\label{ass:invertible}
Let $\Zvec$ be the input process and $\tilde{\Zvec}$ the output process of a static system. If $\hat{\Zvec}$ is another process which is equivalent to $\Zvec$ in the sense that there exists a (not necessarily static) invertible system which converts one to the other, then
\begin{subequations}
 \begin{equation}
 \relLossrate{\Zvec\to\tilde{\Zvec}}=\relLossrate{\hat{\Zvec}\to\tilde{\Zvec}}.
\end{equation}
Likewise, if $\hat{\Zvec}$ is equivalent to $\tilde{\Zvec}$, then
\begin{equation}
 \relLossrate{\Zvec\to\tilde{\Zvec}}=\relLossrate{\Zvec\to\hat{\Zvec}}.
\end{equation}
\end{subequations}
\end{ass}
In particular, a polyphase decomposition or a perfect reconstruction filterbank decomposition of a process is equivalent to the original process in above sense.

In some cases not the total information lost in the system is of interest, but only the portion $W$ which is relevant to the user. Hence, in~\cite{Geiger_Relevant_arXiv}, the notion of \emph{relevant information loss} was introduced as the difference between mutual informations:
\begin{equation}
 \loss[W]{Z\to g(Z)} := \mutinf{W;Z}-\mutinf{W;g(Z)}
\end{equation}

This notion is extended to stochastic processes in
\begin{definition}[Relevant Information Loss Rate]\label{def:relevantrate}
Let $\Wvec$ be a process jointly stationary with $\Zvec$, representing the relevant information content of $\Zvec$. Then, the information loss rate relevant w.r.t. $\Wvec$ induced by processing $\Zvec$ to a jointly stationary process $\tilde{\Zvec}$ is
\begin{equation}
 \lossrate[\Wvec]{\Zvec\to \tilde{\Zvec}} := \mutrate{\Wvec;\Zvec}-\mutrate{\Wvec;\tilde{\Zvec}} \label{eq:relevantrate}
\end{equation}
provided the quantities exist.
\end{definition}

As a specific example, $\Wvec$ might be the sign of $\Zvec$, or $\Zvec$ might be a noisy observation of $\Wvec$. Note further that in Definition~\ref{def:relevantrate} it is not necessary to assume that the system is static; if it was, then $\tilde{\Zvec}=g(\Zvec)$, as in Definition~\ref{def:relativerate}.

Considering the scenario depicted in Fig.~\ref{fig:downsampling}, we will be concerned with linear filters and their effect on the information-carrying processes. If the filter is stable and causal\footnote{In addition to stability and causality, its impulse response must not vanish completely in order that the Paley-Wiener condition is satisfied. We will make this mild assumption throughout the rest of the work.}, its magnitude response satisfies the Paley-Wiener condition (cf.~\cite[p.~215]{Papoulis_Fourier} and~\cite[p.~423]{Papoulis_Probability}):
\begin{equation}
  \frac{1}{2\pi} \int_{-\pi}^\pi \ln |H(\e{\jmath\theta})| d\theta > -\infty.\label{eq:Paley_Wiener}
\end{equation}
For such filters, the following two lemmas can be presented:
\begin{lem}\label{lem:ass_filters}
 Let $\Zvec$ be a stochastic process satisfying Assumption~\ref{ass:processes} and let $H$ be a stable, causal linear filter with input $\Zvec$. Then, the output process $\tilde{\Zvec}$ of the filter satisfies Assumption~\ref{ass:processes}.
\end{lem}

\begin{IEEEproof}
See Appendix~\ref{proof:ass_filters}.
\end{IEEEproof}

\begin{lem}\label{lem:relevant_filters}
 Let $\Wvec$ and $\Zvec$ be two jointly stationary stochastic processes satisfying Assumption~\ref{ass:processes}, and let $H$ be a stable, causal linear filter with input $\Zvec$. Then, for $\tilde{\Zvec}$ being the output process of the filter,
\begin{equation}
 \mutrate{\Wvec;\Zvec} = \mutrate{\Wvec;\tilde{\Zvec}}.
\end{equation}
\end{lem}

\begin{proof}
See Appendix~\ref{proof:relevant_filters}.
\end{proof}

As the previous lemma suggests, the output process of a stable, causal linear filter is equivalent to its input process in the sense of Assumption~\ref{ass:invertible}. Moreover, combining this lemma with Definition~\ref{def:relevantrate} one can see that filtering a process with a stable, causal filter does not destroy information:
\begin{equation}
 \lossrate[\Wvec]{\Zvec\to \tilde{\Zvec}} = \mutrate{\Wvec;\Zvec}-\mutrate{\Wvec;\tilde{\Zvec}} = 0.
\end{equation}

\section{Relative Information Loss in a Downsampler}\label{sec:relativeLoss}
Consider the scenario depicted in Fig.~\ref{fig:downsampling}, where $\Xvec$ satisfies Assumption~\ref{ass:processes}. If the filter $H$ is stable and causal, so does $\tilde{\Xvec}$. To analyze the information loss rate in the downsampling device, we employ the relative information loss rate,
\begin{equation}
 \relLossrate{\tilde{\Xvec}^{(M)}\to\Yvec} = \limn \relLoss{(\tilde{X}^{(M)})_1^n\to Y_1^n}~\label{eq:rellossrate:downsampling}
\end{equation}
where we applied $M$-fold blocking to ensure that the mapping between $(\tilde{X}^{(M)})_1^n$ and $Y_1^n$ is static. Downsampling, $Y_n:=\tilde{X}_{nM}$, is now a projection to a single coordinate, hence~\cite{Geiger_RILPCA_arXiv}
\begin{equation}\label{eq:lossrate:downsampler}
 \relLoss{(\tilde{X}^{(M)})_1^n \to Y_1^n}=\frac{\infodim{(\tilde{X}^{(M)})_1^n|Y_1^n}}{\infodim{(\tilde{X}^{(M)})_1^n}}=\frac{n(M-1)}{nM}.%=\frac{M-1}{M} 
\end{equation}
If the filter $H$ is stable and causal and, thus, has no influence on the information content of the stochastic process, we can use Assumption~\ref{ass:invertible} in $(a)$ below and combine~\eqref{eq:rellossrate:downsampling} with~\eqref{eq:lossrate:downsampler} to
\begin{equation}
 \relLossrate{\Xvec^{(M)}\to\Yvec}\stackrel{(a)}{=}\relLossrate{\tilde{\Xvec}^{(M)}\to\Yvec} = \frac{M-1}{M}.
\end{equation}
The amount of information lost in the decimation system in Fig.~\ref{fig:downsampling} is the same for all stable, causal filters $H$.

The question remains whether an \emph{ideal} anti-aliasing filter can prevent information loss, since it guarantees that the downsampling operation is invertible. To show that the answer to this question is negative, take, for example, the ideal low-pass filter recommended in standard textbooks~\cite[Ch.~4.6]{Oppenheim_Discrete3}:
\begin{equation}
 H(\e{\jmath\theta}) = 
\begin{cases}
	1,& \text{ if } |\theta|<\frac{\pi}{M}\\
	0,& \text{ else}
\end{cases}
\end{equation}
We decompose $\Xvec$ in an $M$-channel filterbank: The $k$-th channel is characterized by analysis and synthesis filters being constant in the frequency band $(k-1)/M\le |\theta|<k/M$ and zero elsewhere. Let $\Yvec_{k}$ be the ($M$-fold downsampled) process in the $k$-th channel --- clearly, $\Yvec\equiv\Yvec_{1}$. It can be shown that every $\Yvec_{k}$ satisfies Assumption~\ref{ass:processes} if $\Xvec$ is Gaussian (cf.~proof of Theorem~\ref{thm:aauseless}). Thus we obtain
\begin{equation}
 \relLossrate{\Xvec^{(M)}\to\Yvec}\stackrel{(a)}{=}\relLossrate{{\Yvec}_{1},\dots,{\Yvec}_{M}\to\Yvec_{1}} = \frac{M-1}{M}
\end{equation}
where the information is again lost in a projection and where $(a)$ is due to Assumption~\ref{ass:invertible} since the filterbank decomposition is invertible. The ideal anti-aliasing low-pass filter prevents information from being lost in the downsampler by \emph{destroying information itself}.

If the filter $H$ is a cascade of a causal, stable filter and of one with a piecewise-constant transfer function (with less trivial intervals as pass-bands), the analysis still holds; Information is either lost in the filter or in the downsampler:
\begin{thm}\label{thm:aauseless}
For a Gaussian process $\Xvec$ satisfying Assumption~\ref{ass:processes}, the relative information loss rate in the decimation system depicted in Fig.~\ref{fig:downsampling} satisfies
\begin{equation}
 \relLossrate{\Xvec^{(M)}\to\Yvec} \ge \frac{M-1}{M}
\end{equation}
for every anti-aliasing filter $H$ with finitely many pass-band intervals.
\end{thm}

\begin{IEEEproof}
 See Appendix~\ref{proof:aauseless}.
\end{IEEEproof}

The reason for this seemingly counter-intuitive result is that, without a specific signal model, the amount of information is not necessarily proportional to the amount of energy by which it is represented: There is no reason to prefer a specific frequency band over another. This in some sense parallels our result on the relative information loss in principal components analysis (PCA), where we showed that PCA cannot reduce the amount of information being lost in reducing the dimensionality of the data~\cite{Geiger_RILPCA_arXiv}.

\section{Relevant Information Loss: Gaussian Case}\label{sec:relevantGauss}
To remove the counter-intuitivity of the previous section, we adapt the signal model: Let $\Svec$ and $\Xvec$ be jointly stationary Gaussian processes with PSDs $\psd{S}$ and $\psd{X}$, respectively, with cross PSD $\psd{SX}$, and which satisfy Assumption~\ref{ass:processes}. The information loss rate relevant w.r.t. $\Svec$ is given by
\begin{equation}
 \lossrate[\Svec^{(M)}]{\Xvec^{(M)}\to \Yvec} = \mutrate{\Svec^{(M)};\Xvec^{(M)}}-\mutrate{\Svec^{(M)};\Yvec} \label{eq:relevantrate:downsampling}
\end{equation}
and measures how much of the information $\Xvec$ conveys about $\Svec$ is lost for each output sample due to downsampling.

While in the general case the filter which minimizes $\lossrate[\Svec^{(M)}]{\Xvec^{(M)}\to \Yvec}$ is hard to find, for this Gaussian signal model the solution is surprisingly intuitive:

\begin{definition}[Optimal Energy Compaction Filter~{\cite[Thm.~4]{Vaidyanathan_OptimalFilterBank}}]\label{def:energycompaction}
 The optimal energy compaction filter $H$ for an $M$-fold downsampler and for a given PSD $\psd{X}$ satisfies 
\begin{equation}
 H(\e{\jmath\theta_l})=\begin{cases}
              1, & \text{ for smallest $l$ s.t. } \forall k: \psdk[l]{X}\ge\psdk{X}\\
			0, &\text{ else}
             \end{cases}
\end{equation}
where $\theta_k:=\frac{\theta-2k\pi}{M}$.
\end{definition}

The energy compaction filter for a given PSD can be constructed easily: The $M$-fold downsampled PSD consists of $M$ aliased components; for each frequency point $\theta\in[-\pi/M,\pi/M]$, at least one of them is maximal. The pass-bands of the energy compaction filter correspond to exactly these maximal components~\cite{Tsatsanis_PCFilterBank,Unser_OptimalFilters,Vaidyanathan_OptimalFilterBank}.

\begin{thm}\label{thm:energycompaction}
Let $\Svec$ and $\Xvec$ be jointly stationary Gaussian processes satisfying Assumption~\ref{ass:processes} and having PSDs $\psd{S}$ and $\psd{X}$. Let further the cross PSD $\psd{SX}$ be such that
\begin{equation}
 \int_{-\pi}^\pi \ln\left(\psd{S}\psd{X}-|\psd{SX}|^2\right) d\theta > -\infty.
\end{equation}
Then, the $M$-aliasing-free energy compaction filter for 
\begin{equation}
 \frac{|\psd{SX}|^2 }{\psd{S}\psd{X}-|\psd{SX}|^2}\label{eq:thmEnergyCompaction}
\end{equation}
 minimizes the information loss rate relevant w.r.t. $\Svec$ in the decimation system depicted in Fig.~\ref{fig:downsampling}.
\end{thm}

\begin{IEEEproof}
See Appendix~\ref{proof:energycompaction}.
\end{IEEEproof}

The condition imposed on the cross PSD ensures that the two-dimensional process $(\Svec,\Xvec)$ is regular in the sense of~\cite{Pinsker_InfoEngl}; in particular, it excludes Gaussian processes being linearly dependent, e.g., where $\Xvec$ is obtained by filtering $\Svec$.

The presented theorem admits an interesting
\begin{cor}\label{cor:chen}
 Let $\Svec$ and $\Nvec$ be independent, jointly stationary Gaussian processes satisfying Assumption~\ref{ass:processes} and having PSDs $\psd{S}$ and $\psd{N}$. Let $X_n=S_n+N_n$. Then, the $M$-aliasing-free energy compaction filter for $\psd{S}/\psd{N}$ minimizes the information loss rate relevant w.r.t. $\Svec$ in the decimation system depicted in Fig.~\ref{fig:downsampling}.
\end{cor}

\begin{IEEEproof}
Due to independence, $\psd{SX}=\psd{S}$ and $\psd{X}=\psd{S}+\psd{N}$.
\end{IEEEproof}

The energy compaction filter minimizing the relevant information loss rate thus maximizes the SNR at each frequency. In particular, since for white Gaussian noise $\Nvec$ the energy compaction filter for $\psd{S}/\psd{N}$ coincides with the energy compaction filter for $\psd{S}$, the filter that lets most of the signal's energy pass aliasing-free is also optimal in terms of information.

Also the energy compaction filter of Theorem~\ref{thm:energycompaction} in some sense maximizes the SNR, if one interprets the numerator of~\eqref{eq:thmEnergyCompaction} as the signal, and the denominator as the noise component.

Corollary~\ref{cor:chen} also connects tightly to~\cite{Chen_NyquistShannon}, in which Chen et al. analyzed the capacity of sub-Nyquist sampled, continuous-time additive Gaussian noise channels with frequency response $H_{\mathrm{channel}}(f)$. They showed that the capacity of the channel depends on the (continuous-time) anti-aliasing filter $H_c(f)$, and that the maximizing filter is the energy compaction filter for $|H_{\mathrm{channel}}(f)|^2/S_N(f)$, where $S_N(f)$ is the PSD of the continuous-time noise process~\cite[Thm.~3]{Chen_NyquistShannon}.

\section{Relevant Information Loss: Non-Gaussian Signal plus Gaussian Noise}\label{sec:relevantGeneral}
Although the result for Gaussian processes is interesting due to its closed form, it is of little practical relevance. In many cases, at least the relevant part of $\Xvec$, the data signal process $\Svec$, is non-Gaussian. We thus drop the restriction that $\Svec$ is Gaussian. For the result presented below, we have to assume a signal-plus-Gaussian-noise model, i.e., we assume that $\Xvec$ is the sum of $\Svec$ and an independent Gaussian noise process $\Nvec$.

One can expect that in this case a closed-form solution for $H$ will not be available. Assuming that $\Svec$ is Gaussian yields an upper bound on the information rate $\mutrate{\Svec^{(M)};\Yvec}$. While this upper bound is of little use for filter design (it does not make sense to maximize an upper bound on the information rate), it can also be shown that the Gaussian assumption provides an upper bound on the relevant information loss rate. To this end, we employ the approach of Plumbley~\cite{Plumbley_TN}, who showed that, with a specific signal model, PCA can be justified from an information-theoretic perspective (cf.~also~\cite{Geiger_Relevant_arXiv}).
\begin{thm}\label{thm:gaussbound}
Let $H$ be stable and causal, let $\Svec$ and $\Nvec$ be independent, jointly stationary and satisfy Assumption~\ref{ass:processes}, and let $X_n=S_n+N_n$. $\Nvec$ is Gaussian, and $\Svec_G$ is Gaussian with the same PSD as $\Svec$. Let $X_{G,n}=S_{G,n}+N_n$, and let $\Yvec_G$ be the corresponding output process of the decimation system, respectively. Then,
\begin{equation}
 \lossrate[\Svec^{(M)}]{\Xvec^{(M)}\to \Yvec} \le \lossrate[\Svec^{(M)}_G]{\Xvec^{(M)}_G\to \Yvec_G}.
\end{equation}
\end{thm}

\begin{IEEEproof}
See Appendix~\ref{proof:gaussbound}.
\end{IEEEproof}

% \textcolor{red}{Cite Lapidoth etc. that the Gaussian input is worst or so.}

A consequence of this theorem is that filter design by energetic considerations, i.e., by considering the PSDs of the signals only, has performance guarantees also in information-theoretic terms. In particular, while the theorem is restricted to stable and causal filters, intuition suggests that a high-order filter in some way should approximate the energy compaction filter from Corollary~\ref{cor:chen}. One has to consider, though, that the filter $H$ optimal in the sense of the upper bound might not coincide with the filter optimal w.r.t. $\lossrate[\Svec^{(M)}]{\Xvec^{(M)}\to \Yvec}$.

Note that, to the best of our knowledge, the statements of Theorem~\ref{thm:gaussbound} cannot be generalized to arbitrary correlations between $\Svec$ and $\Xvec$, as in Theorem~\ref{thm:energycompaction}. The reason is that applying Plumbley's idea requires an independent, additive Gaussian noise component. At best, a generalization to non-Gaussian noise is possible, if the noise is more Gaussian than the signal in a well-defined sense (cf.~\cite{Geiger_Relevant_arXiv}). This generalization, however, is within the scope of future work.

\section{Examples}\label{sec:example}
We now illustrate our results with an example: Let the PSD of $\Svec$ be given by $\psd{S}=1+\cos\theta$ and let $\Nvec$ be independent white Gaussian noise with variance $\sigma^2$, i.e., $\psd{N}=\sigma^2$. The PSD of $\Xvec$ is depicted in Fig.~\ref{fig:psd}. We consider downsampling by a factor of $M=2$. Were $\Svec$ Gaussian too, the optimal filter would be an ideal low-pass filter with cut-off frequency $\pi/2$ (cf.~Corollary~\ref{cor:chen}).

\begin{figure}[t]
   \centering
   \begin{pspicture}[showgrid=false](-4,-.2)(4,2.5)
	\psset{yunit=0.75cm}
	\psaxeslabels[style=Arrow](0,0)(-4,-0.5)(4,2.75){$\theta$}{$\psd{X}$}
	\psplot[style=Graph]{-4}{4}{x 180 mul Pi div cos 3 add 2 div}
	\psTick{90}(3.14,0)\uput[-90](3.14,0){$\pi$}
	\psTick{90}(-3.14,0)\uput[-90](-3.14,0){$-\pi$}
	\psTick{0}(0,1)\uput[0](0,1){$\sigma^2$}
	\end{pspicture}
	\caption{Power spectral density of $\Xvec$.}
	\label{fig:psd}
\end{figure}

If we assume that $\Svec$ is non-Gaussian, Theorem~\ref{thm:gaussbound} allows us to design a finite-order filter which minimizes an upper bound on the relevant information loss rate. In particular, it can be shown that among all first-order FIR filters with impulse response $h[n]=\delta[n]+c\delta[n-1]$, the filter with $c=1$ minimizes the Gaussian bound (see also Section~\ref{sec:FIR}).

\begin{figure}[t]
   \centering
  \begin{pspicture}[showgrid=false](-4,-0.25)(4,4)
	\psset{unit=0.6cm}
 	\small
	\psaxes[Ox=-20,Dx=10,dx=3,Dy=0.4,dy=1]{->}(-6,0)(-6.5,-0.5)(6.5,5.5)[$10\ln(\sigma^2)$,90][$\lossrate[\Svec_G^{(2)}]{\Xvec_G^{(2)}\to \Yvec_G}$,0]
	\rput[lt](-0.25,6){\psframebox%
	{\begin{tabular}{ll}
	  	\psline[linewidth=1pt,linecolor=black](0.1,0.1)(1,0.1) &\hspace*{0.25cm} $\mutrate{\Xvec_G^{(2)};\Svec_G^{(2)}}$\\%
		\psline[linewidth=1pt,linecolor=red](0.1,0.1)(1,0.1) &\hspace*{0.25cm} Ideal Low-Pass\\
		\psline[linewidth=1pt,linecolor=red,style=Dash](0.1,0.1)(1,0.1) &\hspace*{0.25cm} 1st-order Low-Pass\\%
		\psline[linewidth=1pt,linecolor=blue](0.1,0.1)(1,0.1) &\hspace*{0.25cm} No Filter
	 \end{tabular}}}
	\readdata{\Available}{relevantrate_opt1order_available.dat}
	\readdata{\no}{relevantrate_opt1order_nofilt.dat}
	\readdata{\fir}{relevantrate_opt1order_fir.dat}
	\readdata{\opt}{relevantrate_opt1order_optfilt.dat}
	\psset{xunit=1.8cm,yunit=3cm}
	\listplot[xStart=-1,plotstyle=curve,linecolor=black,linewidth=1pt]{\Available}
	\listplot[plotstyle=curve,linecolor=red,style=Dash,linewidth=0.5pt]{\fir}
	\listplot[plotstyle=curve,linecolor=red]{\opt}
	\listplot[plotstyle=curve,linecolor=blue]{\no}
 \end{pspicture}
	\caption{Upper bounds on the relevant information loss rate in nats as a function of the noise variance $\sigma^2$ for various filter options ($M=2$).}
	\label{fig:lossrate}
\end{figure}
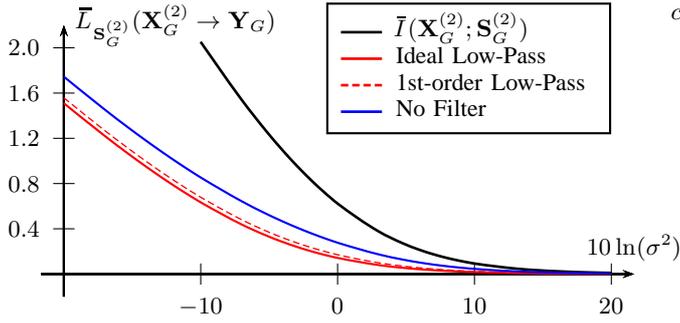

Fig.~\ref{fig:lossrate} shows the upper bound on the relevant information loss rate as a function of the noise variance $\sigma^2$ for the ideal low-pass filter and the optimal first-order FIR filter compared to the case where no filter is used. In addition, the available information $\mutrate{\Xvec_G^{(2)};\Svec_G^{(2)}}=2\mutrate{\Xvec_G;\Svec_G}$ is plotted, which decreases with increasing noise variance. Indeed, filtering can reduce the relevant information loss rate compared to omitting the filter. This is in stark contrast with the results of Section~\ref{sec:relativeLoss}, in which we showed that the relative information loss rate equals $1/2$ regardless of the filter. The reason is that in Section~\ref{sec:relativeLoss} we did not have a signal model in mind, treating every bit of information equally. As soon as one knows which aspect of a stochastic process is relevant, one can successfully apply signal processing methods to retrieve as much information as possible (or to remove as much of the 
irrelevant information as possible, cf.~\cite{Geiger_Relevant_arXiv}).

Interestingly, as Fig.~\ref{fig:lossrate} shows, the improvement of a first-order FIR filter over direct downsampling is significant. Using low-order filters is beneficial also from a computational perspective: To the best of our knowledge, the optimization problem does not permit a closed-form solution for the filter coefficients in general. Thus, numerical procedures will benefit from the fact that the number of coefficients can be kept small. 

\begin{figure}[t]
   \centering
  \begin{pspicture}[showgrid=false](-4,-0.25)(4,6)
	\psset{unit=0.6cm}
 	\small
	\psaxes[Ox=-20,Dx=10,dx=3,Dy=0.4,dy=1]{->}(-6,0)(-6.5,-0.5)(6.5,10)[$10\ln(\sigma^2)$,90][$\lossrate[\Svec_G^{(3)}]{\Xvec_G^{(3)}\to \Yvec_G}$,0]
	\rput[lt](-0.25,8){\psframebox%
	{\begin{tabular}{ll}
	  	\psline[linewidth=1pt,linecolor=black](0.1,0.1)(1,0.1) &\hspace*{0.25cm} $\mutrate{\Xvec_G^{(3)};\Svec_G^{(3)}}$\\%
		\psline[linewidth=1pt,linecolor=red](0.1,0.1)(1,0.1) &\hspace*{0.25cm} Ideal Low-Pass\\
		\psline[linewidth=1pt,linecolor=red,style=Dash](0.1,0.1)(1,0.1) &\hspace*{0.25cm} 1st-order Low-Pass\\%
		\psline[linewidth=1pt,linecolor=red,linestyle=dotted](0.1,0.1)(1,0.1) &\hspace*{0.25cm} 2nd-order Low-Pass\\%
		\psline[linewidth=1pt,linecolor=blue](0.1,0.1)(1,0.1) &\hspace*{0.25cm} No Filter
	 \end{tabular}}}
	\readdata{\Available}{aa_M3_available.dat}
	\readdata{\no}{aa_M3_nofilt.dat}
	\readdata{\firone}{aa_M3_fir1.dat}
	\readdata{\firtwo}{aa_M3_fir2.dat}
	\readdata{\opt}{aa_M3_optfilt.dat}
	\psset{xunit=1.8cm,yunit=4.5cm}
	\listplot[xStart=-1.3,plotstyle=curve,linecolor=black,linewidth=1pt]{\Available}
	\listplot[plotstyle=curve,linecolor=red,style=Dash,linewidth=0.5pt]{\firone}
	\listplot[plotstyle=curve,linecolor=red,linestyle=dotted,linewidth=0.5pt]{\firtwo}
	\listplot[plotstyle=curve,linecolor=red]{\opt}
	\listplot[plotstyle=curve,linecolor=blue]{\no}
 \end{pspicture}
	\caption{Upper bounds on the relevant information loss rate in nats as a function of the noise variance $\sigma^2$ for various filter options ($M=3$). Note that the curve for the optimal second-order FIR filter is not visible in this figure, because it falls together with the curve of the optimal filter.}
	\label{fig:threefold:loss}
\end{figure}
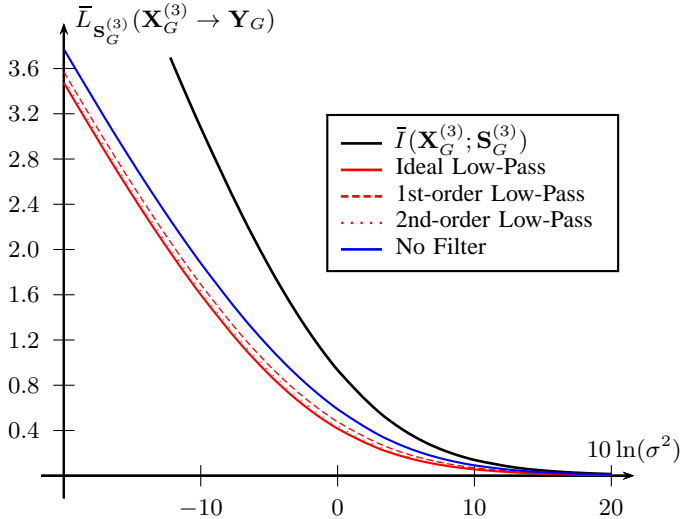

\begin{figure}[t]
   \centering
  \begin{pspicture}[showgrid=false](-4,-0.25)(4,3.5)
	\psset{unit=0.6cm}
 	\small
	\psaxes[Ox=-20,Oy=0,Dx=10,dx=3,Dy=0.02,dy=1.666]{->}(-6,0)(-6.5,-0.5)(6.5,5)[$10\ln(\sigma^2)$,90][$c_1-\sqrt{2}$,180]
	\readdata{\coeff}{aa_M3_coeff2.dat}
	\psset{xunit=1.8cm,yunit=50cm}
	\listplot[plotstyle=curve,linecolor=red]{\coeff}
 \end{pspicture}
	\caption{Difference between the second-order FIR filter coefficient $c_1$ minimizing $\lossrate[\Svec_G^{(3)}]{\Xvec_G^{(3)}\to \Yvec_G}$ and the coefficient maximizing the filter output signal-to-noise ratio ($c_1=\sqrt{2}$).}
	\label{fig:threefold:coeff}
\end{figure}

\begin{figure}[t]
   \centering
  \begin{pspicture}[showgrid=false](-4,-0.25)(4,3.5)
	\psset{unit=0.6cm}
 	\small
	\psaxes[Ox=-20,Dx=10,dx=3,Dy=0.03,dy=1.666]{->}(-6,0)(-6.5,-0.5)(6.5,5)[$10\ln(\sigma^2)$,90][$10^{-3}$,180]
	\readdata{\Loss}{aa_M3_SNRLoss.dat}
	\psset{xunit=1.8cm,yunit=100cm}
	\listplot[plotstyle=curve,linecolor=red]{\Loss}
 \end{pspicture}
	\caption{Difference between the relevant information loss rates ($M=3$) of the second-order FIR filters maximizing the filter output signal-to-noise ratio and minimizing $\lossrate[\Svec_G^{(3)}]{\Xvec_G^{(3)}\to \Yvec_G}$.}
	\label{fig:threefold:addloss}
\end{figure}

We repeated the experiment with the same PSDs but with a three-fold downsampler, i.e., for $M=3$. For the Gaussian assumption, the first-order FIR filter with impulse response $h[n]=\delta[n]+\delta[n-1]$ again proved optimal. Here, however, we also determined numerically the optimal filter coefficients for a second-order FIR filter with impulse response $h[n]=\delta[n]+c_1\delta[n-1]+c_2\delta[n-2]$. Remarkably, for all considered variances, the optimal value for $c_2$ is equal to one. The optimal value for $c_1$, however, depends on the variance $\sigma^2$ of the noise process, as indicated in Fig.~\ref{fig:threefold:coeff}. The filter coefficient is close to $\sqrt{2}$, which yields the impulse response vector equal to the maximal eigenvector of the input process' autocorrelation matrix, and hence to the solution maximizing the filter output signal-to-noise ratio (see Section~\ref{sec:FIR}). While the difference diminishes for large noise variance, for strong signals the coefficient is significantly 
different. This clearly illustrates that energetic and information-theoretic designs are inherently different, and one can hope to have similar solutions to both cost functions only in few, specialized scenarios. Knowing whether such a scenario applies or not is of prime importance for the system designer, since it could admit simple energetic design approaches to circumvent the need for non-linear, non-convex optimization to achieve the information-theoretic optimum.

Comparing the relevant information loss rates depicted in Fig.~\ref{fig:lossrate} and Fig.~\ref{fig:threefold:loss}, one can observe that the loss is greater than for two-fold downsampling. For comparison, again the available information rate $\mutrate{\Xvec_G^{(3)};\Svec_G^{(3)}}=3\mutrate{\Xvec_G;\Svec_G}$ is plotted. Finally, Fig.~\ref{fig:threefold:addloss} shows the additional loss induced by replacing the ideal coefficient $c_1$ by $\sqrt{2}$, the coefficient yielding a maximum output signal-to-noise ratio. As can be seen, the additional loss is negligible, which justifies energetic design considerations from an information-theoretic point-of-view, at least in this example and for strong additive white Gaussian noise.

\section{Discussion and Outlook}\label{sec:FIR}
In our opinion, the present work has two important messages worth repeating: First of all, we showed that, with no signal model in mind, anti-aliasing filtering is futile. \emph{Assuming} that it is sensible to preserve as much energy aliasing-free as possible is guesswork and expresses the common misconception that energy and information behave similarly. In truth, the relevant information may be transported mainly in small signal components, or even in the sum of the alias terms. Hence, it might even turn out that anti-aliasing filtering, as it is proposed by standard textbooks on signal processing, does more harm than good. In this sense, the analysis of anti-aliasing filtering parallels our previous analysis of PCA in~\cite{Geiger_RILPCA_arXiv}, where the conclusion was similar.

The second important message is that with a specific signal model in mind, anti-aliasing filtering can indeed reduce the information loss in the downsampler. In particular, for a Gaussian signal-plus-noise model it does make sense to preserve the \emph{signal} components with the largest signal-to-noise ratio aliasing-free. Then, the information-theoretic optimum coincides with the energetic one, and filter design based on second-order statistics is well-justified.\\

One of the most important aims of future work is the design of finite-order filters with information-theoretic cost functions. While the information-maximizing filter with unconstrained order is simple to obtain (cf.~Theorem~\ref{thm:energycompaction}), the practically more relevant case of finite-order filters is much more difficult even in the purely Gaussian signal-plus-noise case: The problem of maximizing (cf.~Lemma~\ref{lem:relevant_filters})
\begin{multline}\label{eq:FIRline}
  \mutrate{\Svec^{(M)}_G;\Yvec_G}=\mutrate{\tilde{\Svec}^{(M)}_G;\Yvec_G}  \\= \frac{1}{4\pi}\int_{-\pi}^\pi \ln\left(1+
\frac{\sum_{k=0}^{M-1} \psdk{S}|H(\e{\jmath\theta_k})|^2}
{\sum_{k=0}^{M-1} \psdk{N}|H(\e{\jmath\theta_k})|^2} \right) d\theta
\end{multline}
does, except in particularly simple cases (see Section~\ref{sec:example}), not permit a closed-form solution, nor is it necessarily convex.

The situation simplifies when the noise is white, i.e., when $\psd{N}=\sigma_N^2$, and with the restriction that the filter satisfies the Nyquist-$M$ condition~\cite{Kirac_FIR,Vaidyanathan_OptimalFilterBank}
\begin{equation}
 \frac{1}{M}\sum_{k=0}^{M-1} \left|H(\e{\jmath\theta_k}) \right|^2 = 1.\label{eq:nyquist}
\end{equation}
This restriction is meaningful, e.g., when the filter is part of an orthonormal filter bank or a principal component filter bank.

Employing these restrictions and applying Jensen's inequality to~\eqref{eq:FIRline} yields an upper bound on the information rate
\begin{align}
 &\mutrate{\tilde{\Svec}^{(M)}_G;\Yvec_G}\notag  \\
&\le \frac{1}{2}\ln\left(1+\frac{1}{2\pi\sigma_N^2}\int_{-\pi}^\pi \frac{1}{M}\sum_{k=0}^{M-1}\psdk{S}|H(\e{\jmath\theta_k})|^2d\theta\right)\notag\\
&\stackrel{(a)}{=} \frac{1}{2}\ln\left(1+\frac{1}{2\pi\sigma_N^2}\int_{-\pi}^\pi \psd{S}|H(\e{\jmath\theta})|^2d\theta\right)\notag\\
&=\frac{1}{2}\ln\left(1+\frac{\sigma_{\tilde{S}}^2}{\sigma_N^2}\right)=\frac{1}{2}\ln\left(1+\frac{\sigma_{\tilde{S}}^2}{\sigma_{\tilde{N}}^2}\right)\label{eq:firbound}
\end{align}
where $(a)$ is because the variance of a stationary process does not change during downsampling and where $\sigma_{\tilde{S}}^2$ ($\sigma_{\tilde{N}}^2$) is the variance of $\tilde{\Svec}$ ($\tilde{\Nvec}$), the output of $H$ to the input process $\Svec$ ($\Nvec$).

Maximizing an upper bound on the information rate thus amounts to maximizing the signal-to-noise ratio, or equivalently, the signal power, at the output of the downsampler or filter. This is exactly the objective of optimum FIR compaction filters for $\psd{S}$, which have been investigated in~\cite{Kirac_FIR} and the references therein. The solution for filter orders strictly smaller than the downsampling factor $M$ is the maximal eigenvector of the autocorrelation matrix~\cite{Kirac_FIR}. For larger filter orders, various analytical and numerical methods exist; see~\cite{Tuqan_FIR} for an overview. All these represent a sub-optimal solution to the original problem of designing information-maximizing FIR filters; the problem of designing finite-order filters with, e.g., rational transfer functions, remains elusive.

Obviously, the upper bound~\eqref{eq:firbound} is the better the larger the noise variance $\sigma_N^2$ is. Hence, energetic design considerations will succeed especially in cases where the Gaussian noise is white and has a large variance; see also Fig.~\ref{fig:threefold:addloss}. One has to keep in mind, however, that even the problem of FIR filters is solved only sub-optimally, since FIR energy compaction filters only \emph{maximize an upper bound on the information rate}; the desired result, however, is either a lower bound on the information rate or an upper bound on the relevant information loss rate. Future work shall deal with this issue.

The extension of this work's results to sampling of continuous-time processes is also of great interest: The extension in terms of relevant information loss rate has been made partly in~\cite{Chen_NyquistShannon}, presenting a result similar to our Corollary~\ref{cor:chen}. The authors of~\cite{Chen_NyquistShannon} furthermore showed that a filterbank sampling mechanism can have a strictly larger capacity than a single-channel sampling mechanism, suggesting that one can further reduce the relevant information loss rate in the downsampler by replacing the filter $H$ by a filterbank. 

In terms of relative information loss rates, the extension to continuous-time processes is immediate via employing a sampling expansion (Nyquist rate) and successive downsampling, at least for bandlimited processes. If the input process is not bandlimited and has a positive PSD a.e., we conjecture that the relative information loss rate will approach unity, i.e., that 100\% of the available information is lost. 

Finally, the generalization of our Theorems~\ref{thm:aauseless} and~\ref{thm:gaussbound} to non-Gaussian processes and general filters $H$, respectively, is the goal of future work. While the former is already sketched in Appendix~\ref{proof:aauseless}, the latter requires deeper investigation.

% \section{Conclusion}
% In this work we analyzed the information loss in a decimation system as a function of its constituting anti-aliasing filter. In particular, we showed that without a signal model in mind, anti-aliasing filtering is futile since it cannot reduce the information loss even if ideal filters are permitted. The situation changes for a simple signal-plus-Gaussian-noise model, where the information loss w.r.t. the signal process can be reduced by properly choosing the filter. As a direct consequence, we concluded that filter design based on second-order statistics of the process can be justified from an information-theoretic perspective.

\appendices
\section{Proof of Theorem~\ref{thm:aauseless}}\label{proof:aauseless}
\begin{figure*}[t]
\centering
 \begin{pspicture}[showgrid=false](-8,-6)(5,-1)
  \psset{style=RoundCorners,style=Arrow,gratioWh=1.2}
	\pssignal(-8,-2){X}{$\Xvec$}
	\dotnode(-6.5,-2){d1}\dotnode(-6.5,-3){d2}
	\psfblock(-5,-2){H0}{$H_1$}
	\psfblock(-5,-3){H1}{$H_2$}
	\psfblock(-5,-5){HM}{$H_{LM}$}
	\psdsampler(-3,-2){DO1}{$LM$}
	\psdsampler(-3,-3){DO2}{$LM$}
	\psdsampler(-3,-5){DOm}{$LM$}
	\psusampler(0,-2){UP1}{$LM$}
	\psusampler(0,-3){UP2}{$LM$}
	\psusampler(0,-5){UPm}{$LM$}
	\psfblock(2,-2){Hs0}{$H_1$}
	\psfblock(2,-3){Hs1}{$H_2$}
	\psfblock(2,-5){HsM}{$H_{LM}$}
	\pscircleop(3.5,-2){oplus1}
	\pscircleop(3.5,-3){oplus2}
	\pssignal(5,-2){y2}{$\Xvec$}
	\ldotsnode[angle=90](-3,-4){dots0}\ldotsnode[angle=90](0,-4){dots0}
	\ldotsnode[angle=90](2,-4){dots0}\ldotsnode[angle=90](3.5,-4){dots2}
	\ldotsnode[angle=90](-5,-4){dots0}\ldotsnode[angle=90](-6.5,-4){dots1}
	\nclist{ncline}[naput]{X,d1,H0,DO1,UP1 $\tilde{\Yvec}_1$,Hs0,oplus1,y2}
	\nclist{ncline}[naput]{d2,H1,DO2,UP2 $\tilde{\Yvec}_2$,Hs1,oplus2}
	\nclist{ncline}[naput]{HM,DOm,UPm $\tilde{\Yvec}_{LM}$,HsM}
	\nclist{ncline}[nbput]{d1,d2,dots1}
    \nclist{ncline}[nbput]{dots2,oplus2,oplus1}
	\ncangle[angleA=-90,angleB=180]{dots1}{HM} \ncangle[angleB=-90]{HsM}{dots2}
 \end{pspicture}
\caption{Filterbank decomposition of the input process $\Xvec$.}
\label{fig:polyphase}
\end{figure*}
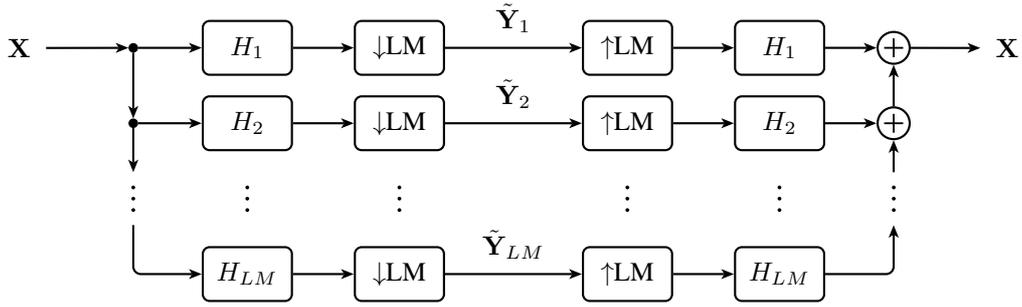

The case $H\equiv 1$ and the case of a stable and causal $H$ have already been dealt with. Thus, assume that $H$ is piecewise constant with $H(\e{\jmath\theta})$ being either one or zero. This assumption is unproblematic, since $H$ can always be split into a filter satisfying this assumption and a set of filters satisfying the Paley-Wiener condition~\eqref{eq:Paley_Wiener}. The latter filters can be omitted as made clear above.

Next, assume that the pass-band and stop-band intervals have rational endpoints. In other words, since there are only finitely many such intervals, there exists an even integer $L$ large enough such that the pass-band interval endpoints are integer multiples of $1/L$. With this in mind, observe Fig.~\ref{fig:polyphase} which illustrates the filterbank decomposition of $\Xvec$~\cite[Ch.~4.7.6,~p.~230]{Oppenheim_Discrete3}. There, $H_i$ is an ideal brick-wall filter for the $i$-th frequency band, i.e., 
\begin{equation}
  H_i(\e{\jmath\theta}) = 
\begin{cases}
	1, &\text{ if } \frac{(i-1)\pi}{LM} < |\theta-2k\pi|\le \frac{i\pi}{LM},\quad k\in\mathbb{Z}\\
    0, &\text{ else}
\end{cases}.\label{eq:brickwallfilter}
\end{equation}
Since $\Xvec$ is a Gaussian process~\cite[p.~663]{Papoulis_Probability},
\begin{equation}
 \derate{\Xvec} = \frac{1}{2}\ln(2\pi\e{}) + \frac{1}{4\pi}\int_{-\pi}^\pi \ln \psd{X} d\theta, \label{eq:gaussbound}
\end{equation}
and from $|\derate{\Xvec}|<\infty$ it follows that $\psd{X}>0$ a.e. It naturally follows that, for all $i=1,\dots,LM$, $\psd{\tilde{Y}_i}>0$ a.e., where of course $\tilde{\Yvec}_i$ is Gaussian too.

Clearly, the variance of the $i$-th downsampled process $\tilde{\Yvec}_i$ is positive (since its PSD is positive a.e.) and finite (since it is upper bounded by $LM$ times the variance of $\Xvec$). Thus, $|\diffent{\tilde{Y}_i}|<\infty$, and $\derate{\tilde{\Yvec}_i}<\infty$. The differential entropy rates of $\tilde{\Yvec}_i$ are obtained by splitting the integral in~\eqref{eq:gaussbound} into $LM$ parts; the sum of these $LM$ parts is $-\infty$ if at least one of its parts is $-\infty$ (since none of these parts can be $\infty$ by the fact that $|\diffent{\tilde{Y}_i}|<\infty$). Thus, $|\derate{\tilde{\Yvec}_i}|<\infty$, and with Lemma~\ref{lem:DEInfoDim}, it follows that
\begin{equation}
 \infodim{(\tilde{Y}_i)_1^n}=n
\end{equation}
for all $i$. Moreover, since the downsampled processes $\tilde{\Yvec}_i$ are mutually independent (they are uncorrelated and jointly Gaussian), it follows that
\begin{equation}
 \infodim{(\tilde{Y}_1)_1^n,\dots,(\tilde{Y}_{LM})_1^n}=nLM.
\end{equation}
Since the collection $\tilde{\Yvec}:=\{\tilde{\Yvec}_1,\dots,\tilde{\Yvec}_{LM}\}$ is equivalent to $\Xvec$, in the sense that perfect reconstruction is possible, by Assumption~\ref{ass:invertible},
\begin{equation}
 \relLossrate{\XvecMb\to \Yvec} = \relLossrate{\tilde{\Yvec}\to \blocked{\Yvec}{L}}.
\end{equation}

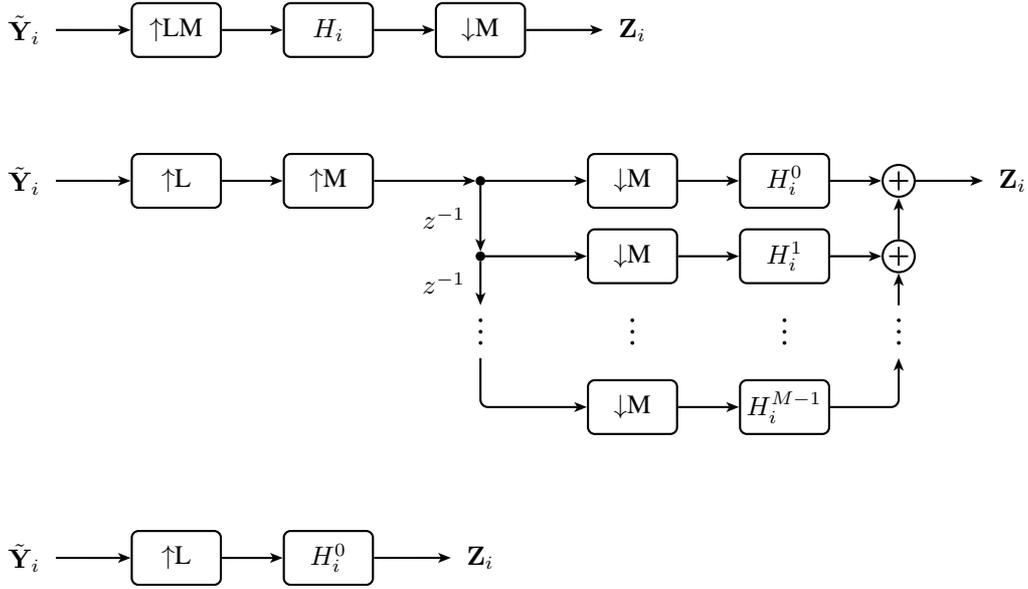
\begin{figure*}[t]
\centering
 \begin{pspicture}[showgrid=false](-8,-8)(5,1)
  \psset{style=RoundCorners,style=Arrow,gratioWh=1.2}
	\pssignal(-8,0){Yi}{$\tilde{\Yvec}_i$}
	\psusampler(-6,0){UP}{$LM$}
	\psfblock(-4,0){H}{$H_i$}
	\psdsampler(-2,0){Do}{$M$}
	\pssignal(0,0){Z}{$\Zvec_i$}
	\nclist{ncline}{Yi,UP,H,Do,Z}

	\pssignal(-8,-2){Y1}{$\tilde{\Yvec}_i$}
	\dotnode(-2,-2){d1}\dotnode(-2,-3){d2}
	\psusampler(-6,-2){UP1}{$L$}\psusampler(-4,-2){UP2}{$M$}
	\psdsampler(0,-2){D1}{$M$}
	\psdsampler(0,-3){D2}{$M$}
	\psdsampler(0,-5){Dm}{$M$}
	\psfblock(2,-2){Hs0}{$H_i^0$}
	\psfblock(2,-3){Hs1}{$H_i^1$}
	\psfblock(2,-5){HsM}{$H_i^{M-1}$}
	\pscircleop(3.5,-2){oplus1}
	\pscircleop(3.5,-3){oplus2}
	\pssignal(5,-2){Z1}{$\Zvec_i$}
	\ldotsnode[angle=90](-2,-4){dots1}\ldotsnode[angle=90](0,-4){dots0}
	\ldotsnode[angle=90](2,-4){dots0}\ldotsnode[angle=90](3.5,-4){dots2}
	\nclist{ncline}[naput]{Y1,UP1,UP2,d1,D1,Hs0,oplus1,Z1}
	\nclist{ncline}[naput]{d2,D2,Hs1,oplus2}
	\nclist{ncline}[naput]{Dm,HsM}
	\nclist{ncline}[nbput]{d1,d2 $z^{-1}$,dots1 $z^{-1}$}
    \nclist{ncline}[nbput]{dots2,oplus2,oplus1}
	\ncangle[angleA=-90,angleB=180]{dots1}{Dm} \ncangle[angleB=-90]{HsM}{dots2}

	\pssignal(-8,-7){Yi3}{$\tilde{\Yvec}_i$}
	\psusampler(-6,-7){UP3}{$L$}
	\psfblock(-4,-7){H3}{$H_i^0$}
	\pssignal(-2,-7){Z3}{$\Zvec_i$}
	\nclist{ncline}{Yi3,UP3,H3,Z3}
 \end{pspicture}
\caption{All three systems are equivalent. The first equivalence is due to the polyphase decomposition of the decimation system $H_i$ followed by the $M$-fold downsampler. The second equivalence is due to the fact that the filters in all but the first branch have an input signal identical to zero. $H_i^0$ is the $M$-fold downsampled filter $H_i$, i.e., it has impulse response $h_i[nM]$. By linearity, $\Xvec$ is the sum of the processes $\Zvec_i$, $i=1,\dots,LM$.}
\label{fig:rearranging}
\end{figure*}

We employ the linearity of the system to move the filter $H$ next to the reconstruction filters $H_i$. By the assumption made about the pass-bands of $H$, the cascade of $H$ and $H_i$ either equals $H_i$ or is identical to zero. The filter $H$ thus amounts to eliminating some of the sub-band processes $\tilde{\Yvec}_i$; a simple projection. What remains to be analyzed is the effect of the $M$-fold downsampler, which can also be moved next to the reconstruction filters due to linearity. Notice that with the polyphase decomposition of decimation systems (cf.~\cite[Ch.~4.7.4,~p.~228]{Oppenheim_Discrete3}), the $i$-th branch of the filterbank can be rearranged as in Fig.~\ref{fig:rearranging}. Due to the cascade of up- and downsampling, only the filter $H_i^0$ is relevant, while all other filters $H_i^l$ will have vanishing input. In particular, while $H_i$ is given by~\eqref{eq:brickwallfilter}, one gets for the filter $H_i^0$ with impulse response $h_i[nM]$,
\begin{align}
   H_i^0(\e{\jmath\theta}) &=
\frac{1}{M}\sum_{m=0}^{M-1}  H_i(\e{\jmath\frac{\theta-2m\pi}{M}})\\
&=
\begin{cases}
	\frac{1}{M}, &\text{ if } \frac{(i-1)\pi}{L} < |\theta-2k\pi|\le \frac{i\pi}{L},\quad k\in\mathbb{Z}\\
    0, &\text{ else}
\end{cases} \label{eq:filtersequal}
\end{align}
where the last line follows from the fact that $H_i$ is $LM$-aliasing-free and, hence, $M$-aliasing free ($H_i$ have bandwidths $1/LM$ and fall in exactly one of the bands with width $1/M$).

By the $2\pi$-periodicity of the transfer functions it follows that the sequence of filters is periodic with $2L$, i.e., $H_i^0 = H_{i+2L}^0$. Moreover, $H_{L+k}^0 = H_{L-k+1}^0$, $k=1,\dots,L$, by the symmetry of the filter. Therefore, there are exactly $L$ different filters, each occurring $M$ times.

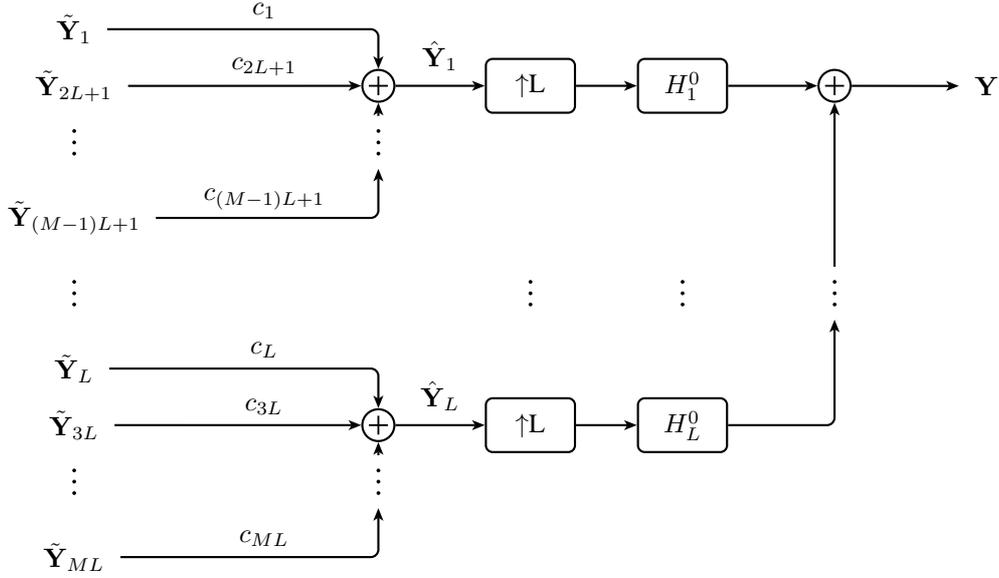
\begin{figure*}[t]
\centering
 \begin{pspicture}[showgrid=false](-9,-8.25)(4,-0.5)
  \psset{style=RoundCorners,style=Arrow,gratioWh=1.2}
	\pssignal(-8,-0.75){Y1}{$\tilde{\Yvec}_1$}
	\pssignal(-8,-1.5){YL1}{$\tilde{\Yvec}_{2L+1}$}
	\ldotsnode[angle=90](-8,-2.25){dots0}\ldotsnode[angle=90](-4,-2.25){dots1}
	\pssignal(-8,-3.25){YL1M}{$\tilde{\Yvec}_{(M-1)L+1}$}
	\pscircleop(-4,-1.5){oplus1}
	\ncline{YL1}{oplus1}  \rput(-5.5,-1.25){$c_{2L+1}$}
	\ncangle[angleB=90]{Y1}{oplus1} \rput(-5.5,-0.5){$c_1$}
	\ncline{dots1}{oplus1} 
	\ncangle[angleB=-90]{YL1M}{dots1} \rput(-5.5,-3){$c_{(M-1)L+1}$}
	
	\ldotsnode[angle=90](-8,-4.25){dots0}

	\pssignal(-8,-5.25){YL}{$\tilde{\Yvec}_L$}
	\pssignal(-8,-6){YLM}{$\tilde{\Yvec}_{3L}$}
	\pssignal(-8,-7.75){YLMM}{$\tilde{\Yvec}_{ML}$}
	\ldotsnode[angle=90](-8,-6.75){dots0}\ldotsnode[angle=90](-4,-6.75){dots2}
	\pscircleop(-4,-6){oplus2}
	\ncline{YLM}{oplus2}  \rput(-5.5,-5.75){$c_{3L}$}
	\ncangle[angleB=90]{YL}{oplus2} \rput(-5.5,-5){$c_L$}
	\ncline{dots2}{oplus2} 
	\ncangle[angleB=-90]{YLMM}{dots2} \rput(-5.5,-7.5){$c_{ML}$}

	\psusampler(-2,-1.5){UP1}{$L$}
	\psfblock(0,-1.5){Hs1}{$H_1^0$}
	\psusampler(-2,-6){UP2}{$L$}
	\psfblock(0,-6){Hs2}{$H_L^0$}
	\ldotsnode[angle=90](-2,-4.25){dots0}
	\ldotsnode[angle=90](0,-4.25){dots0}
	
	\pscircleop(2,-1.5){oplusend}
	\ldotsnode[angle=90](2,-4.25){dotsend}
	\pssignal(4,-1.5){Y}{$\Yvec$}

	\nclist{ncline}[naput]{oplus1,UP1 $\hat{\Yvec}_1$,Hs1,oplusend,Y}
	\nclist{ncline}[naput]{oplus2,UP2 $\hat{\Yvec}_L$,Hs2}
	\ncangle[angleB=-90]{Hs2}{dotsend}
	\ncline{dotsend}{oplusend}
 \end{pspicture}
\caption{Equivalent system for a decimation filter with a piecewise constant $H$ (pass-band intervals with rational endpoints). The constants $c_i$ indicate whether or not the sub-band process is eliminated by $H$, i.e., $c_i\in\{0,1\}$. Note that with~\eqref{eq:filtersequal} the interpolator outputs can be added without information loss. Thus, information loss only occurs by eliminating and/or adding sub-band processes -- a cascade of an invertible linear map and a projection. The system is shown for $M$ odd. Note that $\hat{Y}_i$ depends on $\{\tilde{Y}_i,\tilde{Y}_{2L-i+1},\tilde{Y}_{2L+i},\tilde{Y}_{4L-i+1},\dots\}$.}
\label{fig:final}
\end{figure*}

Combining the last system from Fig.~\ref{fig:rearranging} with Fig.~\ref{fig:polyphase} and~\eqref{eq:filtersequal}, the schematic in Fig.~\ref{fig:final} is obtained. Note that since the filters $H_i^0$ are orthogonal and $L$-aliasing-free (i.e., the frequency response of the filter does not overlap by $L$-fold downsampling and can thus be reconstructed perfectly), adding the reconstruction filter outputs does not incur information loss. We thus again use Assumption~\ref{ass:invertible} and write
\begin{equation}
 \relLossrate{\XvecMb\to\Yvec} = \relLossrate{\tilde{\Yvec}\to\blocked{\Yvec}{L}} = \relLossrate{\tilde{\Yvec}\to\hat{\Yvec}}
\end{equation}
where $\hat{\Yvec}:=\{\hat{\Yvec}_1,\dots,\hat{\Yvec}_L\}$. But the transformation from $\tilde{\Yvec}$ to $\hat{\Yvec}$ is linear, specifically, the cascade of an invertible linear map and a projection. We therefore apply~\cite[Cor.~1]{Geiger_RILPCA_arXiv} and get
\begin{multline}
 \relLoss{(\tilde{Y}_1)_1^n,\dots,(\tilde{Y}_{LM})_1^n \to (\hat{Y}_1)_1^n,\dots,(\hat{Y}_L)_1^n} \\
= 1- \frac{\infodim{(\hat{Y}_1)_1^n,\dots,(\hat{Y}_L)_1^n}}{\infodim{(\tilde{Y}_1)_1^n,\dots,(\tilde{Y}_{LM})_1^n}} = 1- \frac{\infodim{(\hat{Y}_1)_1^n,\dots,(\hat{Y}_L)_1^n}}{nML}.
\end{multline}
The information dimension of $\{(\hat{Y}_1)_1^n,\dots,(\hat{Y}_L)_1^n\}$ is bounded from above by the number of its scalar components, which is $nL$. This completes the proof for filters $H$ with rational endpoints of the pass-band intervals.

Assume now that one of the interval endpoints is an irrational $a_i$. Then, for a fixed $L$, there exists $A_i\in\mathbb{Z}$ such that $A_i/L<a_i<(A_i+1)/L$. Obviously, the filter with the irrational endpoint replaced by either of these two rational endpoints destroys either more or less information (either the corresponding coefficient $c_m$ in Fig.~\ref{fig:final} is zero or one). For both of these filters, however, the information dimension of $\hat{Y}_1^n$ cannot exceed $nL$, and above analysis holds. This completes the proof.
\endproof 

Note that the proof suggests how to measure the exact relative information loss rate for the decimation system by evaluating the information dimension of $\hat{Y}_1^n$. For rational endpoints of the pass-band intervals this is simple since $\infodim{(\hat{Y}_i)_1^n}\in\{0,n\}$. For irrational endpoints one can always wedge the filter $H$ between one with destroys more and one which destroys less information; for $L$ sufficiently large, the resulting difference in the relative information loss rates will be small, and eventually vanish in the limit $L\to\infty$.

Moreover, the result should also hold for non-Gaussian processes satisfying Assumption~\ref{ass:processes}. The intuition behind this is a bottleneck consideration: Since the filterbank decomposition is perfectly invertible, the information dimensions of the input and output processes need to be identical for all time windows $\{1,\dots,n\}$. The Gaussian assumption was required to show that the information dimension (for a given time window) of each sub-band process is related to the information dimension of the input process (in the same time window) and the number of filterbank channels. We believe that the Gaussian assumption can be removed by the fact that all operations in the model are Lipschitz, and that therefore the information dimension cannot increase, cf.~\cite{Wu_Renyi}. As a consequence, it is not possible that the information dimension of the sub-band processes is smaller than in the Gaussian case, since then the information dimension of the (reconstructed) output would be smaller than the information dimension of the input -- a contradiction.

\section{Proof of Theorem~\ref{thm:energycompaction}}\label{proof:energycompaction}
The goal is to maximize the information rate between $\Yvec$ and the $M$-dimensional input process $\Mblocked{\Svec}$, i.e., $\mutrate{\Mblocked{\Svec};\Yvec}$, because it is the only component of $\lossrate[\Mblocked{\Svec}]{\Mblocked{\Xvec}\to\Yvec}$ depending on $H$. 

The $l$-th coordinate of $\Mblocked{\Svec}$ shall be the process $\Svec_l$ with samples $S_{l},S_{M+l},S_{2M+l},\cdots$, where $l=1,\dots,M$. We note in passing that the processes $\Svec_l$ constitute the polyphase decomposition of $\Svec$. The PSD of the $l$-th coordinate is given as\footnote{It is immaterial whether the sum runs from 0 to $M-1$ or from 1 to $M$. We will make repeated use of this fact below.}
\begin{equation}
 \psd{ll}:=\psd{S_l}=\frac{1}{M}\sum_{k=0}^{M-1} \psdk{S};
\end{equation}
the cross PSD between the $l$-th and the $m$-th coordinate is
\begin{equation}
 \psd{lm}:=\psd{S_lS_m}=\frac{1}{M}\sum_{k=0}^{M-1} \psdk{S}\e{\jmath(l-m)\theta_k}.
\end{equation}
Note that $\psd{lm}=\kpsd{ml}$.

With $\psd{\tilde{X}}=\psd{X}|H(\e{\jmath\theta})|^2$ and $Y_n:=\tilde{X}_{nM}$ the PSD of $\Yvec$ is given as
\begin{equation}
 \psd{Y} = \frac{1}{M}\sum_{k=0}^{M-1} \psdk{X}|H(\e{\jmath\theta_k})|^2.\label{eq:psdy}
\end{equation}

Finally, if $\psd{SX}$ is the cross PSD between $\Svec$ and $\Xvec$, one has $\psd{S\tilde{X}}=\psd{SX}H^*(\e{\jmath\theta})$~\cite[Ch.~9-4]{Papoulis_Probability} and for the cross-PSD between $\Svec_l$ and $\Yvec$,
\begin{equation}
 \psd{lY}:=\psd{S_lY} = \frac{1}{M}\sum_{k=0}^{M-1} \psdk{SX} H^*(\e{\jmath\theta_k}) \e{\jmath l\theta_k}.
\end{equation}
Again, $\psd{lY}=\kpsd{Yl}$.

Let $\Amat_S$ be the $M\times M$ PSD matrix containing the elements $\psd{lm}$, let $\svec_Y$ be a column vector with elements $\psd{lY}$, and let
\begin{equation}
 \Amat_{SY}=\left[
\begin{array}{cc}
  \Amat_S & \svec_Y\\ \svec_Y^H & \psd{Y}
\end{array}
\right]
\end{equation}
where $^H$ is the Hermitian transposition. Then, if
\begin{equation}
 \int_{-\pi}^\pi \ln |\det \Amat_S|d\theta > -\infty \label{eq:nregular}
\end{equation}
the information rate $\mutrate{\Mblocked{\Svec};\Yvec}$ equals~\cite[Thm.~10.4.1]{Pinsker_InfoEngl}
\begin{equation}
 \mutrate{\Mblocked{\Svec};\Yvec}=\frac{1}{4\pi}\int_{-\pi}^\pi \ln\frac{ \psd{Y} \det\Amat_S}{\det{\Amat_{SY}}} d\theta.\label{eq:proof:rate}
\end{equation}

To verify that condition~\eqref{eq:nregular} holds, note that $\Amat_S=\Wvec\boldsymbol{\Lambda}\Wvec^H$, where the $(l,m)$-th element of $\Wvec$ is $\e{\jmath l \theta_m}/\sqrt{M}$ and where $\boldsymbol{\Lambda}$ is a diagonal matrix with $\psdk[l]{S}$ in its $l$-th position. Hence, $\Wvec$ is unitary, $\det\Amat_S=\det\boldsymbol{\Lambda}=\prod_{l=0}^{M-1}\psdk[l]{S}$, and, doing some calculus,
\begin{equation}
 \frac{1}{4\pi}\int_{-\pi}^\pi \ln |\det \Amat_S|d\theta = M\derate{\Svec}.
\end{equation}

We now consider the fraction $\det\Amat_{SY}/\det\Amat_S=\det\Amat_S^{-1}\det\Amat_{SY}$. According to Cauchy's expansion~\cite[p.~26]{Horn_Matrix},
\begin{equation}
 \det\Amat_{SY} = \psd{Y}\det{\Amat_S}-\svec_Y^H\mathrm{adj}\Amat_S \svec_Y.
\end{equation}
Since $\Amat_S$ is non-singular a.e. by Assumption~\ref{ass:processes}, we can write for the adjugate $\mathrm{adj}\Amat_S=\Amat_S^{-1}\det \Amat_S$. Hence,
\begin{equation}
 \frac{\det\Amat_{SY}}{\det{\Amat_S}} = \psd{Y}-\svec_Y^H \Amat_S^{-1} \svec_Y.
\end{equation}
With $\Amat_S^{-1}=\Wvec\boldsymbol{\Lambda}^{-1}\Wvec^H$, we can write for
\begin{align}
 &(\Amat_S^{-1} \svec_Y)_l\notag \\
&= \sum_{n=1}^{M} (\Amat_S^{-1})_{ln} \psd{nY}\\
&= \sum_{n,k,m=1}^{M} \frac{\psdk[m]{SX} H^*(\e{\jmath\theta_m})}{M^2\psdk{S}}\e{\jmath(l-n)\theta_k} \e{\jmath n\theta_m}.
\end{align}
But
\begin{equation}
 \sum_{n=0}^{M-1} \e{\jmath n (\theta_m-\theta_k)}= \sum_{n=0}^{M-1} \e{\jmath n (m-k)\frac{2\pi}{M}}
\end{equation}
vanishes if $m\neq k$ and evaluates to $M$ otherwise. Thus, 
\begin{equation}
 (\Amat_S^{-1} \svec_Y)_l
= \sum_{k=1}^{M} \frac{\psdk{SX} H^*(\e{\jmath\theta_k})}{M\psdk{S}}\e{\jmath l \theta_k}.
\end{equation}
Finally,
\begin{align}
 &\svec_Y^H \Amat_S^{-1} \svec_Y\notag \\
&= \sum_{l=1}^{M} \kpsd{lY} (\Amat_S^{-1} \svec_Y)_l \\
&= \sum_{k,l,m=1}^{M} \frac{\psdk{SX} H^*(\e{\jmath\theta_k})}{M^2\psdk{S}}\kpsdk[m]{SX} H(\e{\jmath\theta_m}) \e{\jmath l(\theta_k-\theta_m)} .
\end{align}
Summing $\e{\jmath l (\theta_k-\theta_m)}$ over index $l$ again vanishes for $k\neq m$ and evaluates to $M$ otherwise. Hence,
\begin{equation}
 \svec_Y^H \Amat_S^{-1} \svec_Y= \frac{1}{M} \sum_{k=0}^{M-1} \frac{|\psdk{SX}|^2 |H(\e{\jmath\theta_k})|^2}{\psdk{S}}.
\end{equation}
Maximizing the information rate amounts to pointwise maximizing the argument of the logarithm of~\eqref{eq:proof:rate}, i.e.,
\begin{align}
 & \frac{\sum_{k=0}^{M-1} \psdk{X}|H(\e{\jmath\theta_k})|^2}{\sum_{k=0}^{M-1} \psdk{X}|H(\e{\jmath\theta_k})|^2 - \frac{|\psdk{SX}|^2 |H(\e{\jmath\theta_k})|^2}{\psdk{S}}}\notag\\
&= 1 + \frac{\sum_{k=0}^{M-1}\frac{|\psdk{SX}|^2 }{\alpha(\theta_k)\psdk{S}}\alpha(\theta_k)|H(\e{\jmath\theta_k})|^2}{\sum_{k=0}^{M-1} \alpha(\theta_k) |H(\e{\jmath\theta_k})|^2}\label{eq:proof:second}
\end{align}
where we inserted~\eqref{eq:psdy} for $\psd{Y}$ and where
\begin{equation}
 \alpha(\theta_k):=\frac{\psdk{S}\psdk{X}-|\psdk{SX}|^2}{\psdk{S}}.
\end{equation}
The second term in~\eqref{eq:proof:second} is a weighted average with weights $w_k(\theta):=\alpha(\theta_k)|H(\e{\jmath\theta_k})|^2/\sum_{k'=0}^{M-1} \alpha(\theta_{k'}) |H(\e{\jmath\theta_{k'}})|^2$:
\begin{equation}
 1 + \sum_{k=0}^{M-1}\frac{|\psdk{SX}|^2 }{\psdk{S}\psdk{X}-|\psdk{SX}|^2}w_k(\theta)
\end{equation}
The maximum is achieved by setting $w_k(\theta)=1$ for the first index $k$ satisfying, for all $l=0,\dots,M-1$,
\begin{multline}
 \frac{|\psdk{SX}|^2 }{\psdk{S}\psdk{X}-|\psdk{SX}|^2} \\\ge \frac{|\psdk[l]{SX}|^2 }{\psdk[l]{S}\psdk[l]{X}-|\psdk[l]{SX}|^2}.
\end{multline}
Evidently, all other weights have to be set to zero.

The filter $H$ is thus related to the piecewise constant functions $w_k(\theta)$ via
\begin{equation}
 |H(\e{\jmath\theta_k})|^2 = \frac{\psdk{S}}{\psdk{S}\psdk{X}-|\psdk{SX}|^2}w_k(\theta)
\end{equation}
where the relation has to be fulfilled for all $k=0,\dots,M-1$. By assumption, the denominator corresponds to the squared magnitude response of a causal, stable filter, and since Lemma~\ref{lem:relevant_filters} holds, one can choose $H$ to be piecewise constant. That $H$ is identical to the optimal energy compaction filter for $\frac{|\psd{SX}|^2 }{\psd{S}\psd{X}-|\psd{SX}|^2}$ is evident from Definition~\ref{def:energycompaction}.\endproof

\section{Proof of Theorem~\ref{thm:gaussbound}}\label{proof:gaussbound}
Note that with Lemma~\ref{lem:relevant_filters}
\begin{equation}
 \lossrate[\Mblocked{\Svec}]{\Mblocked{\Xvec}\to\Yvec} = \lossrate[\tilde{\Svec}^{(M)}]{\tilde{\Xvec}^{(M)}\to\Yvec}
\end{equation}
where $\tilde{\Xvec}$ ($\tilde{\Svec}$) is obtained by filtering $\Xvec$ ($\Svec$) with $H$. Since $\tilde{X}_n=\tilde{S}_n+\tilde{N}_n$, and since $Y_n=\tilde{X}_{nM}$, one obtains
\begin{align}
&\lossrate[\tilde{\Svec}^{(M)}]{\tilde{\Xvec}^{(M)}\to\Yvec} \notag\\
&= \limn \frac{1}{n} \left(\diffent{\tilde{X}_1^{nM}}-\diffent{\tilde{X}_1^{nM}|\tilde{S}_1^{nM}}\right.\notag\\
&\quad \left.-\diffent{Y_1^n}+\diffent{Y_1^n|\tilde{S}_1^{nM}}\right)\\
&= \limn \frac{1}{n} \left(\diffent{\tilde{X}_1^{nM}}-\diffent{\tilde{X}_M,\dots,\tilde{X}_{nM}}\right.\notag\\
&\quad\left.-\diffent{\tilde{N}_1^{nM}}+\diffent{\tilde{N}_M,\dots,\tilde{N}_{nM}}\right)\\
&= \limn \frac{1}{n} \diffent{\tilde X_1^{M-1},\dots,\tilde X_{(n-1)M+1}^{nM-1}|\tilde X_M,\dots,\tilde X_{nM}}\notag\\
&\quad - \limn \frac{1}{n} \diffent{\tilde{N}_1^{nM}}+\limn \frac{1}{n} \diffent{\tilde{N}_M,\dots,\tilde{N}_{nM}}.
\end{align}
The conditional differential entropy in the last equation is always upper bounded by the corresponding expression for RVs $\tilde{X}_{G,1},\dots,\tilde{X}_{G,mM}$ with the same joint first- and second-order moments as the original RVs~\cite[Thm.~8.6.5, p.~254]{Cover_Information2}. Replacing $\Svec$ by $\Svec_G$ yields $\Xvec_G$ and $\Yvec_G$ Gaussian (by Gaussianity of $\Nvec$) and achieves this upper bound with equality. Taking the limit completes the proof.\endproof

\section{Auxiliary Results}\label{app:auxiliary}
\subsection{Proof of Lemma~\ref{lem:DEInfoDim}}\label{proof:DEInfoDim}
By assumption, $|\diffent{Z}|<\infty$ and $|\derate{\Zvec}|<\infty$. But since $\derate{\Zvec}=\limn \diffent{Z_n|Z_1^{n-1}}$  it also follows from conditioning~\cite[Cor. to Thm.~8.6.1, p.~253]{Cover_Information2} and the chain rule of differential entropy~\cite[Thm.~8.6.2, p.~253]{Cover_Information2} that
% \begin{equation}
%  -\infty < \derate{\Zvec} \le \diffent{Z_n|Z_1^{n-1}} \le \diffent{Z} < \infty
% \end{equation}
% and similarly, by 
% \begin{equation}
%  -\infty < n\derate{\Zvec} \le \diffent{Z_1^n} \le n\diffent{Z} < \infty.
% \end{equation}
% Now take $Z_\LetterSet{J}:=\{Z_j : j\in\LetterSet{J}\}$. Clearly, by conditioning and the chain rule,
\begin{equation}
 -\infty < \card{\LetterSet{J}} \derate{\Zvec} \le \diffent{Z_\LetterSet{J}}\le \card{\LetterSet{J}}\diffent{Z} <\infty.
\end{equation}
By the assumption of finite variance, $\ent{\lfloor Z\rfloor}<\infty$ and the information dimension of $Z$ exists~\cite[Prop.~1]{Wu_Renyi} (and similarly for any finite collection $Z_\LetterSet{J}$ of samples). But $\diffent{Z_\LetterSet{J}}$ is the $\card{\LetterSet{J}}$-dimensional entropy of $Z_\LetterSet{J}$, which can only be finite\footnote{If $\infodim{Z}<d$, the $d$-dimensional entropy of $Z$ would be $-\infty$; if $\infodim{Z}>d$, the $d$-dimensional entropy of $Z$ would be $\infty$.} if $\infodim{Z_\LetterSet{J}}=\card{\LetterSet{J}}$~\cite{Renyi_InfoDim}. This completes the proof.\endproof

\subsection{Proof of Lemma~\ref{lem:ass_filters}}\label{proof:ass_filters}
If the filter is stable (i.e., the impulse response is absolutely summable~\cite[Ch.~2.4, p.~59]{Oppenheim_Discrete3} and, thus, square summable) and causal, the Paley-Wiener theorem~\cite[p.~215]{Papoulis_Fourier} states that
\begin{equation}
 \frac{1}{2\pi} \int_{-\pi}^\pi \ln |H(\e{\jmath\theta})| d\theta > -\infty.
\end{equation}
Moreover, since the filter is stable, one has by Jensen's inequality:
\begin{align}
 \frac{1}{4\pi} \int_{-\pi}^\pi \ln |H(\e{\jmath\theta})|^2 d\theta
&= \frac{1}{2}\expec{\ln |H(\e{\jmath\theta})|^2}\\
&\le \frac{1}{2} \ln \left(\expec{|H(\e{\jmath\theta})|^2}\right)\\
&= \frac{1}{2} \ln \NoiseGain<\infty
\end{align}
where the expectation is taken assuming the frequency variable $\theta$ is uniformly distributed on $[-\pi, \pi]$ and where the noise gain $G$ is
\begin{equation}
 \NoiseGain := \frac{1}{2\pi} \int_{-\pi}^\pi |H(\e{\jmath\theta})|^2 d\theta.
\end{equation}
The last (strict) inequality follows by assuming stability (square summability of the impulse response and Parseval's theorem~\cite[Tab.~2.2, p.~86]{Oppenheim_Discrete3}).

According to~\cite[p.~663]{Papoulis_Probability}, the differential entropy rate at the output of the filter $H$ is given by
\begin{equation}
 \derate{\tilde{\Zvec}} = \derate{\Zvec} + \frac{1}{2\pi} \int_{-\pi}^\pi \ln |H(\e{\jmath\theta})| d\theta
\end{equation}
and thus, by assumption, finite.

With~\cite[(9-190), p.~421]{Papoulis_Probability}, one has for the autocorrelation function of the output of a linear filter
\begin{equation}
 r_{\tilde{Z}\tilde{Z}}[m] = \left(r_{ZZ}*\rho\right) [m]
\end{equation}
where
\begin{equation}
 \rho[m] = \sum_{k} h[m+k]h^*[k].
\end{equation}
Thus, by the fact that for a zero-mean process the variance satisfies $\sigma^2=r_{ZZ}[0]\ge |r_{ZZ}[m]|$,
\begin{align}
 r_{\tilde{Z}\tilde{Z}}[0] &\le \sigma^2 \sum_m |\rho[m]|\\
&\le \sigma^2 \sum_m\sum_{k} |h[m+k]h^*[k]|\\
&= \sigma^2 \sum_{k}\left( |h^*[k]|\sum_{m} |h[m+k]|\right)\\
&\le \sigma^2 \sum_{k}\left( |h^*[k]| C \right)\\
&\le \sigma^2 C^2 < \infty
\end{align}
where the last two lines follow from stability of $H$ (the impulse response is absolutely summable) and by the assumption that $\Zvec$ has finite variance. Thus, by conditioning and the maximum-entropy property of the Gaussian distribution,
\[
 -\infty<\derate{\tilde{\Zvec}}\le \diffent{\tilde{Z}} \le \frac{1}{2}\ln(2\pi\e{}\sigma^2 C^2) <\infty.
\]
This completes the proof.\endproof 

\subsection{Proof of Lemma~\ref{lem:relevant_filters}}\label{proof:relevant_filters}
The proof is provided for jointly Gaussian processes $\Wvec$ and $\Zvec$ only; since the effect of linear filters on the differential entropy rate is independent of the process statistics (cf.~\cite[p.~663]{Papoulis_Probability}), the result can be extended to the general case.

First, note that a stable, causal filter satisfies the Paley-Wiener condition, and that thus $H(\e{\jmath\theta})>0$ a.e. From~\cite[Cor. to Thm.~9-4, p.~412]{Papoulis_Probability} one gets $\psd{\tilde{Z}}=|H(\e{\jmath\theta})|^2 \psd{Z}$. Since $\Zvec$ has a finite entropy rate, $\psd{Z}>0$ a.e., and, thus, $\psd{\tilde{Z}}>0$ a.e. That for the cross PSD $\psd{\tilde{Z}W}=H(\e{\jmath\theta})\psd{ZW}$ holds can be shown easily.

From~\cite[Thm.~10.2.1,~p.~175]{Pinsker_InfoEngl},
\begin{equation}
 \mutrate{\Zvec;\Wvec} = -\frac{1}{4\pi}\int_{-\pi}^\pi \ln\left(1-|\rho_{ZW}(\e{\jmath\theta})|^2\right)
\end{equation}
where
\begin{equation}
 |\rho_{ZW}(\e{\jmath\theta})|^2 
=\begin{cases}
	\frac{|\psd{ZW}|^2}{\psd{Z}\psd{W}}, &\text{ if } \psd{ZW}\neq0\\
	0, &\text{ else}
\end{cases}.
\end{equation}
With above reasoning one gets
\begin{multline}
 |\rho_{\tilde{Z}W}(\e{\jmath\theta})|^2 \\
=\begin{cases}
	\frac{|H(\e{\jmath\theta})|^2|\psd{ZW}|^2}{|H(\e{\jmath\theta})|^2\psd{Z}\psd{W}}, &\text{ if } H(\e{\jmath\theta})\psd{ZW}\neq0\\
	0, &\text{ else}
\end{cases}
\end{multline}
which is a.e. equal to $|\rho_{ZW}(\e{\jmath\theta})|^2$ since $H(\e{\jmath\theta})>0$ a.e. This completes the proof.\endproof

\bibliographystyle{IEEEtran}
\bibliography{IEEEabrv,/afs/spsc.tugraz.at/project/IT4SP/1_d/Papers/InformationProcessing.bib,%
/afs/spsc.tugraz.at/project/IT4SP/1_d/Papers/ProbabilityPapers.bib,%
/afs/spsc.tugraz.at/user/bgeiger/includes/textbooks.bib,%
/afs/spsc.tugraz.at/user/bgeiger/includes/myOwn.bib,%
/afs/spsc.tugraz.at/user/bgeiger/includes/UWB.bib,%
/afs/spsc.tugraz.at/project/IT4SP/1_d/Papers/InformationWaves.bib,%
/afs/spsc.tugraz.at/project/IT4SP/1_d/Papers/ITBasics.bib,%
/afs/spsc.tugraz.at/project/IT4SP/1_d/Papers/HMMRate.bib,%
/afs/spsc.tugraz.at/project/IT4SP/1_d/Papers/SignalProcessing.bib,%
/afs/spsc.tugraz.at/project/IT4SP/1_d/Papers/ITAlgos.bib}

\end{document}

%% file: abbrevations_processing.tex
% Signals...
\newcommand{\x}[1]{x[#1]}
\newcommand{\y}[1]{y[#1]}
\newcommand{\qX}[1]{\hat{X}^{(#1)}}
\newcommand{\qS}[1]{\hat{S}^{(#1)}}
\newcommand{\qY}[1]{\hat{Y}^{(#1)}}
\newcommand{\Xt}[1]{\tilde{X}_{#1}}
\newcommand{\Yt}[1]{\tilde{Y}_{#1}}
\newcommand{\Si}[1]{S_{#1}}
\newcommand{\dft}[1]{F^{(#1)}}
\newcommand{\Xe}{X_{\epsilon}}

% Entropies
\newcommand{\overbar}[1]{\mkern 1.5mu\overline{\mkern-3mu#1\mkern-0.5mu}\mkern 1.5mu}
\newcommand{\ent}[1]{H(#1)}
\newcommand{\diffent}[1]{h(#1)}
\newcommand{\derate}[1]{\overbar{h}\left(\mathbf{#1}\right)}
\newcommand{\mutinf}[1]{I(#1)}
\newcommand{\ginf}[1]{I_G(#1)}
\newcommand{\kld}[2]{D(#1||#2)}
\newcommand{\kldr}[2]{\bar{D}(#1||#2)}
\newcommand{\kldrate}[2]{\bar{D}(\mathbf{#1}||\mathbf{#2})}
\newcommand{\binent}[1]{H_2(#1)}
\newcommand{\binentneg}[1]{H_2^{-1}\left(#1\right)}
\newcommand{\entrate}[1]{\overbar{H}(\mathbf{#1})}
\newcommand{\mutrate}[1]{\overbar{I}({\mathbf{#1}})}
\newcommand{\redrate}[1]{\bar{R}(\mathbf{#1})}
\newcommand{\pinrate}[1]{\vec{I}(\mathbf{#1})}
\newcommand{\loss}[2][\empty]{\ifthenelse{\equal{#1}{\empty}}{L(#2)}{L_{#1}(#2)}}
\newcommand{\lossrate}[2][\empty]{\ifthenelse{\equal{#1}{\empty}}{\overbar{L}(\mathbf{#2})}{\overbar{L}_{\mathbf{#1}}(\mathbf{#2})}}
\newcommand{\gain}[1]{G(#1)}
\newcommand{\atten}[1]{A(#1)}
\newcommand{\relLoss}[2][\empty]{\ifthenelse{\equal{#1}{\empty}}{l(#2)}{l_{#1}(#2)}}
\newcommand{\relLossrate}[1]{l(\mathbf{#1})}
\newcommand{\relTrans}[1]{t(#1)}
\newcommand{\partEnt}[2]{H^{#1}(#2)}
\newcommand{\redr}[1]{\bar{R}(#1)}
\newcommand{\entr}[1]{\bar{H}(#1)}

% Domains and Sets
\newcommand{\dom}[1]{\mathcal{#1}}
\newcommand{\indset}[1]{\mathbb{I}\left({#1}\right)}
\newcommand{\LetterSet}[1]{\mathbb{#1}}
\newcommand{\suchthat}{|}

% Distributions...
\newcommand{\unif}[2]{\mathcal{U}\left(#1,#2\right)}
\newcommand{\chis}[1]{\chi^2\left(#1\right)}
\newcommand{\chir}[1]{\chi\left(#1\right)}
\newcommand{\normdist}[2]{\mathcal{N}\left(#1,#2\right)}
\newcommand{\Proba}{\mathrm{Pr}}
\newcommand{\Prob}[1]{\Proba(#1)}
\newcommand{\Mar}[1]{\mathrm{Mar}(#1)}
\newcommand{\Qfunc}[1]{Q\left(#1\right)}
\newcommand{\pmeas}[1]{P_{#1}}
\newcommand{\lebesgue}{\lambda}
\newcommand{\pdf}[1]{f_{#1}}
\newcommand{\pmf}[1]{p_{#1}}
\newcommand{\cdf}[1]{F_{#1}}
\newcommand{\borel}[1]{\mathfrak{B}_{#1}}
\newcommand{\invmeas}{\mu}

% Functions...
% \newcommand{\expec}[1]{\mathrm{E}\left\{#1\right\}}
\newcommand{\expec}[1]{\mathbb{E}\left(#1\right)}
\newcommand{\expecwrt}[2]{\mathbb{E}_{#1}\left(#2\right)}
\newcommand{\var}[1]{\mathrm{Var}\left\{#1\right\}}
\renewcommand{\det}{\mathrm{det}}
\newcommand{\cov}[1]{\mathrm{Cov}\left\{#1\right\}}
\newcommand{\sgn}[1]{\mathrm{sgn}\left(#1\right)}
\newcommand{\sinc}[1]{\mathrm{sinc}\left(#1\right)}
\newcommand{\e}[1]{\mathrm{e}^{#1}}
\newcommand{\multint}{\iint{\cdots}\int}
\newcommand{\modd}[3]{((#1))_{#2}^{#3}}
\newcommand{\quant}[1]{Q\left(#1\right)}
\newcommand{\card}[1]{\mathrm{card}(#1)}
\newcommand{\diam}[1]{\mathrm{diam}(#1)}
\newcommand{\rec}[1]{r(#1)}
\newcommand{\recmap}[1]{r_{\mathrm{MAP}}(#1)}
\newcommand{\recsub}[1]{r_{\mathrm{sub}}(#1)}
\renewcommand{\Re}{\mathfrak{R}}
\renewcommand{\Im}{\mathfrak{I}}
\newcommand{\LumpingFunction}{g}
\newcommand{\Iverson}[1]{\left[#1\right]}
\newcommand{\DFT}[1]{\mathsf{DFT}\left\{#1\right\}}
\newcommand{\IDFT}[1]{\mathsf{DFT}^{-1}\left\{#1\right\}}

% Vectors and Matrices
\newcommand{\ivec}{\mathbf{i}}
\newcommand{\hvec}{\mathbf{h}}
\newcommand{\gvec}{\mathbf{g}}
\newcommand{\avec}{\mathbf{a}}
\newcommand{\kvec}{\mathbf{k}}
\newcommand{\fvec}{\mathbf{f}}
\newcommand{\vvec}{\mathbf{v}}
\newcommand{\xvec}{\mathbf{x}}
\newcommand{\zvec}{\mathbf{z}}
\newcommand{\Xvec}{\mathbf{X}}
\newcommand{\Zvec}{\mathbf{Z}}
\newcommand{\Xhvec}{\hat{\mathbf{X}}}
\newcommand{\xhvec}{\hat{\mathbf{x}}}
\newcommand{\xtvec}{\tilde{\mathbf{x}}}
\newcommand{\Yvec}{\mathbf{Y}}
\newcommand{\yvec}{\mathbf{y}}
\newcommand{\Svec}{\mathbf{S}}
\newcommand{\svec}{\mathbf{s}}
\newcommand{\Nvec}{\mathbf{N}}
\newcommand{\Pvec}{\mathbf{P}}
\newcommand{\Rvec}{\mathbf{R}}
\newcommand{\muvec}{\boldsymbol{\mu}}
\newcommand{\wvec}{\mathbf{w}}
\newcommand{\Wvec}{\mathbf{W}}
\newcommand{\Xmat}{\underline{\mathbf{X}}}
\newcommand{\Ymat}{\underline{\mathbf{Y}}}
\newcommand{\ymat}{\underline{\yvec}}
\newcommand{\Hmat}{\mathbf{H}}
\newcommand{\Amat}{\mathbf{A}}
\newcommand{\Fmat}{\mathbf{F}}
\newcommand{\Wh}{\underline{\hat{\mathbf{W}}}}
\newcommand{\Sh}{\underline{\hat{\boldsymbol{\Sigma}}}}
\newcommand{\Qvec}{\mathbf{Q}}
\newcommand{\nuvec}{\boldsymbol{\nu}}
\newcommand{\Vvec}{\mathbf{V}}
\newcommand{\Uvec}{\mathbf{U}}
\newcommand{\pivec}{\boldsymbol{\pi}}

\newcommand{\zerovec}{\mathbf{0}}
\newcommand{\eye}{\mathbf{I}}
\newcommand{\evec}{\mathbf{i}}

\newcommand{\zeroone}{\left[\begin{array}{c}\zerovec^T\\ \eye\end{array} \right]}
\newcommand{\zerooneT}{\left[\begin{array}{cc}\zerovec & \eye\end{array} \right]}
\newcommand{\zerooneM}{\left[\begin{array}{cc}\zerovec &\zerovec^T\\\zerovec& \eye\end{array} \right]}

\newcommand{\covmat}[1]{\underline{\mathbf{C}}_{#1}}
\newcommand{\hcovmat}[1]{\underline{\hat{\mathbf{C}}}_{#1}}
\newcommand{\Cxx}{\mathbf{C}_{XX}}
\newcommand{\Cx}{\covmat{X}}
\newcommand{\Chx}{\hcovmat{X}}
\newcommand{\Cy}{\covmat{Y}}
\newcommand{\Cz}{\mathbf{C}_{\Zvec}}
\newcommand{\Cn}{\mathbf{C}_{\mathbf{N}}}
\newcommand{\Cnt}{\underline{\mathbf{C}}_{\tilde{\mathbf{N}}}}
\newcommand{\Cntm}{\underline{\mathbf{C}}_{\tilde{\mathbf{N}}}}
\newcommand{\Cxh}{\mathbf{C}_{\hat{X}\hat{X}}}
\newcommand{\rxx}{\mathbf{r}_{XX}}
\newcommand{\Cxy}{\mathbf{C}_{XY}}
\newcommand{\Cyy}{\mathbf{C}_{YY}}
\newcommand{\Cnn}{\mathbf{C}_{NN}}
\newcommand{\Cyx}{\mathbf{C}_{YX}}
\newcommand{\Cygx}{\mathbf{C}_{Y|X}}
\newcommand{\Wmat}{\underline{\mathbf{W}}}

\newcommand{\Jac}[2]{\mathcal{J}_{#1}(#2)}

% Other stuff
\newcommand{\NN}{{N{\times}N}}
\newcommand{\perr}{P_e}
\newcommand{\perh}{\hat{\perr}}
\newcommand{\pert}{\tilde{\perr}}

% Index
% \newcommand{\vecind}[1]{\mathbf{#1}}
\newcommand{\vecind}[1]{#1_0^n}
\newcommand{\roots}[2]{{#1}_{#2}^{(i_{#2})}}
\newcommand{\rootx}[1]{x_{#1}^{(i)}}
\newcommand{\rootn}[2]{x_{#1}^{#2,(i)}}

% Abbrevations
\newcommand{\markkern}[1]{f_M(#1)}
\newcommand{\pole}{a_1}
\newcommand{\preim}[1]{g^{-1}[#1]}
\newcommand{\preimV}[1]{\mathbf{g}^{-1}[#1]}
\newcommand{\Xmax}{\bar{X}}
\newcommand{\Xmin}{\underbar{X}}
\newcommand{\xmax}{x_{\max}}
\newcommand{\xmin}{x_{\min}}
\newcommand{\limn}{\lim_{n\to\infty}}
\newcommand{\limk}{\lim_{k\to\infty}}
\newcommand{\limX}{\lim_{\hat{\Xvec}\to\Xvec}}
\newcommand{\limx}{\lim_{\hat{X}\to X}}
\newcommand{\limXo}{\lim_{\hat{X}_1\to X_1}}
\newcommand{\sumin}{\sum_{i=1}^n}
\newcommand{\finv}{f_\mathrm{inv}}%f_{X_n}^{-1}
\newcommand{\ejtheta}{\e{\jmath\theta}}
\newcommand{\khat}{\bar{k}}
\newcommand{\modeq}[1]{g(#1)}
\newcommand{\partit}[1]{\mathcal{P}_{#1}}
\newcommand{\psd}[1]{S_{#1}(\e{\jmath \theta})}
\newcommand{\kpsd}[1]{S^*_{#1}(\e{\jmath \theta})}
\newcommand{\psdk}[2][\empty]{\ifthenelse{\equal{#1}{\empty}}{S_{#2}(\e{\jmath \theta_k})}{S_{#2}(\e{\jmath \theta_{#1}})}}
\newcommand{\kpsdk}[2][\empty]{\ifthenelse{\equal{#1}{\empty}}{S^*_{#2}(\e{\jmath \theta_k})}{S^*_{#2}(\e{\jmath \theta_{#1}})}}

\newcommand{\infodim}[1]{d(#1)}
\newcommand{\boxdim}[1]{d_B(#1)}
\newcommand{\partElement}[2]{\hat{\dom{X}}^{(#1)}_{#2}}
\newcommand{\partcard}{L}
\newcommand{\powerset}[1]{\mathfrak{P}(#1)}
\newcommand{\NoiseGain}{G}
\newcommand{\Lumping}{(\Pvec,\LumpingFunction)}
\newcommand{\Kappa}{\mathcal{K}}
\newcommand{\Iff}{\Leftrightarrow}
\newcommand{\Then}{\Rightarrow}
\newcommand{\ClassSingleEntry}{\mathsf{SE}}
\newcommand{\ClassSFS}[1]{\mathsf{SFS}(#1)}
\newcommand{\ClassKMarkov}{\mathsf{HMC}(k)}
\newcommand{\Set}[1]{\left\{#1\right\}}
\newcommand{\InstantsNonOverLappingTraversal}[3]{\mathcal{N}_{#1}^{\,#2}(#3)}
\newcommand{\InstantsOccupation}[3]{\mathcal{O}_{#1}^{\,#2}(#3)}
\newcommand{\InstantsTraversal}[3]{\mathcal{T}_{#1}^{\,#2}(#3)}
\newcommand{\InstantsHiddenTraversal}{\mathcal{I}}
\newcommand{\RealisationHiddenTraversal}{I}

\newcommand{\IntegerSet}[1]{{[#1]}}

% signal blocks
\newcommand{\delay}[2]{\psblock(#1){#2}{\footnotesize$z^{-1}$}}
\newcommand{\Quant}[2]{\psblock(#1){#2}{\footnotesize$\quant{\cdot}$}}
\newcommand{\moddev}[2]{\psblock(#1){#2}{\footnotesize$\modeq{\cdot}$}}